\documentclass[twoside,11pt]{article}
\usepackage{jair, rawfonts}

\firstpageno{1}

\usepackage[english]{babel}
\usepackage[utf8]{inputenc}
\usepackage{graphicx}

\usepackage{caption}
\usepackage{subcaption}

\newtheorem{theorem}{Theorem}[section]
\newtheorem{proposition}[theorem]{Proposition}

\newtheorem{definition}[theorem]{Definition}

\newtheorem{exmp}[theorem]{Example}

\newenvironment{proof}{\noindent {\it Proof:\ }}{\hspace*{\fill} \qed\\[0.3cm]}
\def\qed{\hbox{${\vcenter{\vbox{
        \hrule height 0.4pt\hbox{\vrule width 0.4pt height 6pt
        \kern5pt\vrule width 0.4pt}\hrule height 0.4pt}}}$}} 
\newenvironment{example}{\begin{exmp} \rm  }{\hspace*{\fill} \qed \end{exmp}}

\newcommand{\uC}[2]{\ensuremath{\begin{pmatrix} #1 \\ #2 \end{pmatrix}}}
\newcommand{\bC}[4]{\ensuremath{\begin{pmatrix} #1 & #2 \\ #3 & #4 \end{pmatrix}}}

\usepackage{amsmath}

\makeatletter
\def\maketag@@@#1{\hbox{\m@th\normalfont\normalsize#1}}
\makeatother

\usepackage{mathtools}

\usepackage{amssymb}
\renewcommand{\bar}{\overline}
\newcommand{\encourages}{\Leftarrow}
\newcommand{\gain}{\mathop{\mathrm{gain}}}
\newcommand{\Span}{\mathop{\mathrm{span}}}
\newcommand{\sgn}{\mathop{\mathrm{sgn}}}

\usepackage{apacite}

\begin{document}

\title{Representing Fitness Landscapes by Valued Constraints\\to Understand the Complexity of Local Search}

\author{%
\name Artem Kaznatcheev \email kaznatcheev.artem@gmail.com\\
\addr Department of Computer Science\\University of Oxford, UK\\ 
\AND
\name David~A. Cohen \email d.cohen@rhul.ac.uk \\
\addr Department of Computer Science\\
Royal Holloway, University of London, UK\\
\AND
\name Peter~G. Jeavons \email peter.jeavons@cs.ox.ac.uk \\
\addr Department of Computer Science\\University of Oxford, UK
}

\maketitle

\begin{abstract}
Local search is widely used to solve combinatorial optimisation problems and to model biological evolution, 
but the performance of local search algorithms on different kinds of fitness landscapes is poorly understood.  
Here we consider how fitness landscapes can be represented using valued constraints, and investigate what the structure of 
such representations reveals about the complexity of local search.

First, we show that for fitness landscapes representable by binary Boolean valued constraints there is a minimal necessary constraint graph that can be easily computed. 
Second, we consider landscapes as equivalent if they allow the same (improving) local search moves; we show that a minimal constraint graph still exists, but is NP-hard to compute.

We then develop several techniques to bound the 
length of any sequence of local search moves.
We show that such a bound can be obtained from the numerical values of the constraints
in the representation,
and show how this bound may be tightened by considering equivalent representations.
In the binary Boolean case, we prove that a degree 2 or tree-structured constraint graph gives a 
quadratic bound on the number of improving moves made by any local search;
hence, any landscape that can be represented by such a model will be tractable for any form of local search.

Finally, we build two families of examples to show that the conditions in our tractability results are essential.
With domain size three, even just a path of binary constraints can model a landscape with an exponentially long sequence of improving moves.
With a treewidth-two constraint graph, even with a maximum degree of three, binary Boolean constraints can model a landscape with an exponentially long sequence of improving moves.
\end{abstract}

\section{Introduction}

Local search techniques are widely used to solve combinatorial optimisation problems,
and have been intensively studied since the 
1980's~\cite{Aaronson2006,Chapdelaine2005,PLS,Llewellyn1989,Malan2013,Ochoa2018,SY91,Tayarani-Najaran2014}.
They have also played a central role in the theory of biological
evolution, ever since Sewall Wright~\citeyear{Wright1932} introduced the idea of viewing the evolution of 
populations of organisms as a local search process over a space of possible genotypes 
with associated fitness values that became known as a ``fitness landscape''.

The term {\em fitness landscape} is now used to designate any structure $(A,f,N)$ 
consisting of a set of points $A$,
a \textit{fitness function} $f$ defined on those points, and a \textit{neighbourhood function} $N$ on those points,
that indicates which pairs of points are sufficiently close to be considered neighbours.
A point $x$ is said to be \emph{locally optimal} if all neighbours are non-improving 
(i.e., $\forall y \in N(x) \; f(x) \geq f(y)$) 
and \emph{globally optimal} if all points are non-improving.
The \emph{local search problem} for a fitness landscape is to find such a local optimum.
We say the problem is solved by a \emph{local search algorithm} if the only moves allowed in the procedure are from a point $x$ to a point $x'$ with $x' \in N(x)$ and $f(x') > f(x)$.  

Many approaches have been developed to try to distinguish fitness landscapes 
where a local or global optimal point can be found efficiently by local search from those where such optimal points cannot be found efficiently. 
In the 1980's and 90's these attempts focused on statistical measures such as correlation between function values at various distances and various notions of {\em ruggedness}~\cite{Malan2013}.
But, by the late 90's there were several studies highlighting the existence of fitness landscapes that were not rugged and yet were hard to optimise. 
Several new approaches have been developed recently, but the performance of local search algorithms on many kinds of fitness landscapes is still poorly understood~\cite{Malan2013,Ochoa2018,Tayarani-Najaran2014}.

In almost all common examples of local search problems, 
the points in the fitness landscape are tuples of values from some domain, and neighbourhoods are defined by some appropriate 
notion of distance between tuples.  
The fitness function in such a landscape can be defined 
by a collection of \emph{valued constraints}, and hence the search problem 
can be translated into a standard \emph{valued constraint satisfaction problem} (VCSP)~\cite{Cohen2013algebraic,Carbonnel2018,Farnqvist12,Kolmogorov2013,strimbu2019,Thapper2015,Thapper2016}.  

In this paper we begin the development of a novel approach to understanding local search on 
fitness landscapes, based on representing the fitness function using valued constraints
and studying the properties of these representations.
Using the VCSP framework allows us to classify fitness landscapes in new ways, and hence to distinguish new classes of fitness landscapes with specific properties.

Finding a locally optimal point on an arbitrary fitness landscape is a complete problem for the class of problems known as polynomial local search (PLS)~\cite{PLS,SY91,Chapdelaine2005}.
This means that it is expected to be computationally intractable in general to find such a local optimum. 
In particular, there exist known constructions to produce families of fitness landscapes 
where every sequence of improving moves to a local optimum 
from some starting points is exponentially long~\cite{SY91}.
On such landscapes, from such points, 
any local-search algorithm will require an exponentially long sequence of 
improving moves to reach a local optimum.

\ShortHeadings{Representing Fitness Landscapes by Valued Constraints}{Kaznatcheev, Cohen, \& Jeavons}

A key goal, therefore, is to identify classes of fitness landscapes where finding a local optimum
is \emph{tractable} (i.e., solvable in polynomial-time).
We do this by identifying classes of fitness landscapes where every sequence of improving moves from every point is at most polynomially long. These classes are based on properties
of the VCSP representation, such as the numerical values of the constraints or the
structure of the constraint graph.

We start by showing in Section~\ref{sec:mageq} that for some classes of fitness landscapes it is possible to efficiently compute a unique minimal representation as a VCSP instance (Theorems~\ref{thm:simplerep} and \ref{thm:valmin}), giving a convenient normal form.
Then in Section~\ref{sec:signeq} we equate all fitness landscapes that have the same improving local search moves (Definition~\ref{def:signequiv}); 
we show that in some important cases a unique minimal representation for each equivalence class still exists (Theorems~\ref{thm:trimprep} and~\ref{thm:signmin}), but can be NP-hard to compute (Theorem~\ref{thm:weakiscoNPcomplete}).

Using these tools, we then develop several techniques to bound the length of any sequence of local search moves, based on properties of the VCSP representation. 
In Section~\ref{sec:span} we show that such a bound can be obtained from the numerical values of 
the constraints in the representation (Proposition~\ref{prop:spanPath}), and show how this bound may be tightened by considering equivalent representations (Examples~\ref{ex:smoothSpanMinimised} and \ref{ex:cycle}).
In Section~\ref{sec:tree-structured}, we prove that fitness landscapes that can be represented by binary Boolean VCSPs with \emph{tree-structured constraint graphs} can have only quadratically long sequences of improving moves (Theorem~\ref{thm:main}) -- hence they are tractable for any local search algorithm. 
Finally, in Section~\ref{sec:robustness}, we give examples of fitness landscapes that have very simple representations, but have exponentially long sequences of improving moves (Examples~\ref{ex:counting} and \ref{ex:countBool}). 

Because our results are based on bounding the length of all possible sequences of improving moves, they apply to all possible local search algorithms, and hence are particularly useful for investigating properties of biological evolution, as we discuss in Section~\ref{sec:bio}.

\section{Background, Notation, and General Definitions}
We will model the points, $A$, in our fitness landscapes 
as assignments to a collection of $n$ variables, indexed by the set $[n] = 1,2,\ldots,n$, 
with domains $D_1,\ldots,D_n$. 
Hence each point corresponds to a vector $x \in D_1 \times \cdots \times D_n$.
We will generally focus on uniform domains (i.e., cases where $D = D_1 = \cdots = D_n$), 
where this simplifies to $x \in D^n$.
In particular, we will often be interested in Boolean domains, where $x \in \{0,1\}^n$, so each
point can be seen as a bit-vector.

The restriction of a variable assignment $x$ to some subset of variables, with indices in 
a set $S \subseteq [n]$, will be denoted $x[S]$, so $x[S] \in \prod_{j \in S} D_j$.
To reference the assignment to the variable at position $i$, we will usually write $x_i$ unless it is ambiguous, in which case we'll use the more general notation $x[i]$.
If we want to modify $x$ by changing a single variable, say the variable at position $i$,
to some element $b \in D_i$, then we'll write $x[i \mapsto b]$.

Given a set of points, $A$, a {\bf fitness function} on $A$ is 
defined to be an integer-valued function defined on $A$, that is, a function  $f: A \rightarrow \mathbb{Z}$. Because we are modelling fitness, rather than cost,  we \textit{maximise} our objective functions in this paper. All results can be carried over directly to the minimisation context.

To complete the definition of a fitness landscape, we will define a {\bf neighbourhood function} on the set of points $A$ to be a function $N: A \rightarrow 2^{A}$.
For simplicity, we will assume this function is symmetric in the sense that if $y \in N(x)$, then $x \in N(y)$, and we will call such a pair $x$ and $y$ {\bf adjacent} points.
Throughout the paper, we will focus on the case where the set of points $A$ is the set of assignments $D_1 \times \cdots \times D_n$ and $N$ is the {\bf 1-flip neighbourhood} defined by $y \in N(x)$ if and only if there is a variable position $i$ such that $x_i \neq y_i$ and this is the only difference (i.e., $\forall j \neq i \quad x_j = y_j$).
Hence, in the case of the Boolean domain, the graph of the function $N$, where the edges are the pairs of adjacent points, will be the $n$-dimensional hypercube.
Note that considering larger neighbourhoods will only increase (or keep the same) the length of the longest ascending path through the fitness landscape.
\begin{definition}[\shortciteNP{dVPK09,CGB13}]
\label{def:fitnessgraph}
Given any fitness landscape $(A,f,N)$, the corresponding \emph{\bf fitness graph} $G$ 
has vertex set $V(G) = A$ and directed edge set $E(G) = \{(x,y) \; | \; y \in N(x) \text{ and } f(y) > f(x)\}$. \end{definition}
The edges of the fitness graph consist of all pairs of adjacent points which have
distinct values of the fitness function, and are oriented from the lower value of the fitness
function to the higher value;
such directed edges represent the possible moves that can be made by a local search algorithm.

A {\bf (valued) constraint} with scope $S \subseteq [n]$ 
is a function $C_S: \prod_{j \in S} D_j \rightarrow \mathbb{Z}$. 
The {\bf arity of a constraint} $C_S$ is the size $|S|$ of its scope.
For unary and binary constraints we will generally omit the set notation and just write $C_i$ for $C_{\{i\}}$ 
or $C_{ij}$ for $C_{\{i,j\}}$, where $i < j$. 
We will represent the values taken by a unary constraint $C_i$ for each domain element 
by an integer vector of length $|D_i|$, and represent the values taken by 
a binary constraint $C_{ij}$ for each pair of domain elements by an integer matrix,
where $x_i$ selects the row and $x_j$ selects the column.
A zero-valued constraint (of any arity) will be denoted by 0.

\begin{definition}
\label{def:VCSP}
An instance of the \emph{valued constraint satisfaction problem (VCSP)} is a set of constraints 
$\mathcal{C} = \{C_{S_1},\ldots,C_{S_m}\}$.
We say that a VCSP instance $\mathcal{C}$ \emph{\bf implements} a fitness function
$f$ if $f(x) = \sum_{k = 1}^m C_{S_k}(x[S_k])$.
\end{definition}
The arity of a VCSP instance is the maximum arity over its constraints;
if this maximum arity is $2$, then we will call it a \emph{binary} VCSP instance.
The instance-size of a VCSP instance is the number of bits needed to specify 
$n$, $m$ and each constraint.

Given any VCSP instance $\mathcal{C}$, we can take $A$ as the set of all possible assignments, 
$f$ as the fitness function implemented by $\mathcal{C}$, and $N$ as the 1-flip neighbourhood, to obtain
an associated fitness landscape, $(A,f,N)$,
and hence an associated fitness graph, $G_{\mathcal{C}}$,
by Definition~\ref{def:fitnessgraph}.
The vertex set of $G_{\mathcal{C}}$ is the set of possible assignments, $A$, 
and hence is exponential in the size of the instance, $\mathcal{C}$, in general.
Each \emph{binary} VCSP instance also has an associated {\em constraint graph}, defined as follows, 
whose vertex set is polynomial in the size of the instance:
\begin{definition}
Given any binary VCSP instance $\mathcal{C}$, the corresponding \emph{\bf constraint graph} has 
vertices $V(\mathcal{C}) = [n]$, 
edges $E(\mathcal{C}) = \{\{i,j\} \; |  \; C_{ij} \in \mathcal{C}, C_{ij} \neq 0\}$, and constraint-neighbourhood function 
$N_\mathcal{C}(i) = \{j \; | \; \{i,j\} \in E(\mathcal{C}) \}$.
\label{def:conG}
\end{definition}

Each fitness landscape has a unique associated fitness graph which specifies 
all possible improving moves that can be made by a local search on that landscape.  
On the other hand, for each fitness landscape there may be a number of different 
VCSP instances that implement the fitness function of that landscape, and they may 
have different constraint graphs. This motivates the search for canonical, minimal or normalised representations of a given landscape, which we explore in the next two sections.

\section{Magnitude-Equivalence}
\label{sec:mageq}

It is clear from Definition~\ref{def:VCSP} that 
different VCSP instances can implement the same fitness function.
\begin{example}
Consider the two small VCSP instances shown in Figure~\ref{fig:magequiv}.
\begin{figure}[htb]
\begin{center}
\includegraphics[width = 0.5\textwidth]{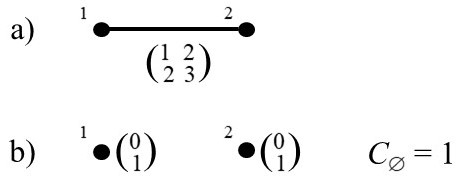}
\end{center}
\vspace{-0.5cm}
\caption{Two VCSP instances implementing the same fitness function} 
\label{fig:magequiv}
\end{figure}

Although these two instances have different constraint graphs,
they both implement the fitness function 
$[f(00),f(01),f(10),f(11)] = [1, 2, 2, 3]$.
\end{example}
We capture this equivalence with the following definition:
\begin{definition}
If two VCSP instances $\mathcal{C}_1$ and $\mathcal{C}_2$ implement the same fitness function $f$,
then we will say they are \emph{\bf magnitude-equivalent}.
\end{definition}
We will show in this section that for binary Boolean VCSP instances 
each equivalence class of magnitude-equivalent VCSP instances has a 
\emph{normal form}: a 
unique, minimal, and easy to compute representative member with special properties.
\begin{definition}
\label{def:simple}
A binary Boolean VCSP instance $\mathcal{C}$ is \emph{\bf simple} if every unary constraint has the form $C_{i} = \uC{0}{c_{i}}$ and every binary constraint has the form $C_{ij} = \bC{0}{0}{0}{c_{ij}}$.
\end{definition}
Each value $c_i$ or $c_{ij}$ will be referred to as the \emph{\bf weight}
of the corresponding constraint.

We now give a direct proof of the following simplification result which is analogous to similar results using constraint propagation in the standard VCSP~\cite{Cooper07:SAC}
(and is essentially a translation of standard results for pseudo-Boolean functions,~\citeNP{CramaHammer2011}).
\begin{theorem}
Any binary Boolean VCSP instance $\mathcal{C}'$ can be transformed into a unique simple VCSP instance
$\mathcal{C}$ that is magnitude-equivalent to $\mathcal{C}'$. 
Moreover, $\mathcal{C}$ can be constructed from $\mathcal{C}'$ in linear time. %
\label{thm:simplerep}
\end{theorem}
\begin{proof}
First, two key observations:
\begin{enumerate}
\item Any unary Boolean constraint $C'_i : \{0,1\} \rightarrow \mathbb{Z}$ can be rewritten as a linear function:
\begin{equation*}
g'_i(x) = (1 - x_i)C'_i(0) + x_iC'_i(1)
\end{equation*}
\item Any binary Boolean constraint $C'_{ij} : \{0,1\} \times \{0,1\} \rightarrow \mathbb{Z}$ can be rewritten as a multilinear polynomial of degree $2$: 
\begin{align*}
g'_{ij}(x) = (1 - x_i)(1 - x_j)C'_{ij}(0,0) & + (1 - x_i)x_j C'_{ij}(0,1) \nonumber \\ 
& + x_i(1 - x_j)C'_{ij}(1,0) + x_ix_jC'_{ij}(1,1).
\end{align*}
\end{enumerate}
From this, we can simplify $\mathcal{C'}$ just by simplifying polynomials:
\begin{align}
    f(x) & = C'_\emptyset + \sum_{i = 1}^n C'_i(x_i) + \smashoperator[lr]{\sum_{\{i,j\} \in E(\mathcal{C'})}} C'_{ij}(x_i,x_j) \nonumber \\
    & = C'_\emptyset + \sum_{i = 1}^n g'_i(x) + \smashoperator[lr]{\sum_{\{i,j\} \in E(\mathcal{C'})}} g'_{ij}(x) \label{eq:badpoly}\\
    & = C_\emptyset + \sum_{i = 1}^n x_ic_i + \smashoperator[lr]{\sum_{1 \leq i < j \leq n}} x_ix_jc_{ij} \label{eq:goodpoly}
\end{align}
where we note that Equation~\ref{eq:badpoly} is a sum of a constant, some linear functions, and some multilinear polynomials of degree 2, and is thus itself a multilinear polynomial of degree 2 (or less).
Equation~\ref{eq:goodpoly} follows from Equation~\ref{eq:badpoly} by multiplying out into monomials and then grouping the coefficients of each similar monomial.
In particular, this gives us the following coefficients:
\begin{align}
    c_i & = C'_{i}(1) - C'_{i}(0) \quad + \sum_{j \; | \; \{i,j\} \in E(\mathcal{C}')}C'_{ij}(1,0) - C'_{ij}(0,0) \label{eq:simple_ci}\\
    c_{ij} & = C'_{ij}(0,0) - C'_{ij}(0,1) - C'_{ij}(1,0) + C'_{ij}(1,1) \label{eq:simple_cij}
\end{align}
The above calculation can be done in linear time.
As the last step in the simplification, 
note that Equation~\ref{eq:goodpoly} corresponds to a VCSP instance $\mathcal{C}$ comprising a nullary constraint $C_\emptyset$, unary constraints $C_i = \uC{0}{c_i}$, and binary constraints $C_{ij} = \bC{0}{0}{0}{c_{ij}}$.
\end{proof}
The next result shows that a simple VCSP instance 
has the minimal constraint graph of any binary instance that implements the same fitness function:
\begin{theorem}
\label{thm:valmin}
Let $\mathcal{C}$ be a simple binary Boolean VCSP instance.
If the binary Boolean VCSP instance $\mathcal{C'}$ is magnitude-equivalent to $\mathcal{C}$, then 
$E(\mathcal{C}) \subseteq E(\mathcal{C'})$.
\end{theorem}
\begin{proof}
Let $e_i \in \{0,1\}^n$ be a variable assignment that sets the $i$th variable to one, and all other variables to zero.
Similarly, let $e_{ij} \in \{0,1\}^n$ be a variable assignment that sets the $i$th and $j$th variables to one, and all other variables to zero.
Let $f$ be the fitness function implemented by $\mathcal{C}$.
Since $\mathcal{C}$ is simple, we have:
\begin{equation*}
    f(e_{ij}) - f(e_i) - f(e_j) + f(0^n) = c_{ij}
\end{equation*}
where we take $c_{ij} = 0$ if $\{i,j\} \notin E(\mathcal{C})$.
Similarly, if $\mathcal{C}'$ also implements $f$, we have:
\begin{equation*}
    f(e_{ij}) - f(e_i) - f(e_j) + f(0^n) = C'_{ij}(1,1) - C'_{ij}(1,0) - C'_{ij}(0,1) + C'_{ij}(0,0)
\end{equation*}
If $\{i,j\} \in E(\mathcal{C})$ then $c_{ij} \neq 0$, so 
$C'_{ij}(1,1) - C'_{ij}(1,0) - C'_{ij}(0,1) + C'_{ij}(0,0) \neq 0$ 
and hence $\{i,j\} \in E(\mathcal{C}')$.
\end{proof}

\section{Sign-Equivalence}
\label{sec:signeq}

In the previous section we considered the equivalence class of all VCSP instances which implement 
precisely the same fitness function. However, when investigating the performance of local search 
algorithms, the exact values of the fitness function are not always relevant; it may be sufficient to 
consider only the fitness graph. 

For example, consider a fitness function $f$, implemented by a VCSP instance $\mathcal{C}$,
where all fitness values are distinct,
but there is at least one pair $i,j$ of positions with no constraint $C_{ij}$. 
Now consider the new fitness function $f'(x) = 2f(x) + C_{ij}(x_i,x_j)$ where $C_{ij} = [0,0;0,1]$.
The fitness graph corresponding to $f'$ is unchanged 
(since all fitness values given by $2f(x)$ differ by 
at least 2, every edge is still present in the fitness graph, 
and no orientations are changed by the new constraint), but
we cannot eliminate this new $C_{ij}$ constraint without changing the precise values of the fitness function at some points.
To capture this similarity between $f$ and $f'$, we introduce a more abstract equivalence relation:
\begin{definition}
\label{def:signequiv}
If two VCSP instances $\mathcal{C}_1$ and $\mathcal{C}_2$ give rise to the same fitness graph,
then we will say they are \emph{\bf sign-equivalent}.
\end{definition}
As with magnitude-equivalence, we will show that for binary Boolean VCSP instances
it is possible to define a normal form or minimal representative member of each equivalence class of sign-equivalent VCSP instances with a unique minimal constraint graph.
Unfortunately, we will see that, unlike the situation for magnitude-equivalence, 
this minimum sign-equivalent constraint graph is NP-hard to compute.
\begin{definition}
In a Boolean fitness graph $G$ with vertex set $\{0,1\}^n$, 
we will say that $i$ \textbf{sign-depends} on $j$ 
if there exists an assignment $x \in \{0,1\}^n$ such that:
\begin{equation*}
    (x,x[i \mapsto \bar{x}_i]) \in E(G) 
    \quad \text{but} \quad 
    (x[j \mapsto \bar{x}_j],x[i \mapsto \bar{x}_i,j \mapsto \bar{x}_j]) \not\in E(G)
\end{equation*}
\label{def:signDep}
\end{definition}
Note that $i$ \textbf{sign-depends} on $j$ if and only if, for any fitness function $f$ that corresponds to the fitness graph $G$, there exists $x \in \{0,1\}^n$ such that:
\begin{equation}
\label{eq:signdep2}
    \text{sgn}(f(x[i \mapsto \bar{x}_i]) - f(x)) \neq \sgn(f(x[i \mapsto \bar{x}_i,j \mapsto \bar{x}_j]) - f(x[j \mapsto \bar{x}_j])).
\end{equation}
We will say that $i$ and $j$ \textbf{sign-interact} if $i$ sign-depends on $j$, 
or $j$ sign-depends on $i$ (or both).
If $i$ and $j$ do not sign-interact then we will say that they are \textbf{sign-independent}.

\begin{definition}
\label{def:trim}
A simple binary Boolean VCSP instance $\mathcal{C}$ with associated fitness graph $G_\mathcal{C}$ 
is called  \emph{\bf trim} if for all $\{i,j\} \in E(\mathcal{C})$, 
$i$ and $j$ sign-interact in $G_\mathcal{C}$.
\end{definition}
Our next result is the sign-equivalence analog of Theorem~\ref{thm:simplerep}, and guarantees a normal form:
\begin{theorem}
\label{thm:trimprep}
Any simple binary Boolean VCSP instance $\mathcal{C}'$ can be transformed into a trim VCSP instance $\mathcal{C}$ that is sign-equivalent to $\mathcal{C}'$.
\end{theorem}
To prove Theorem~\ref{thm:trimprep} we now establish two propositions: 
Proposition~\ref{prop:strongBig} connects the magnitude of constraints with their effect on fitness graphs, 
and Proposition~\ref{prop:magSignDep} connects the magnitude of constraints to sign-interaction.
\begin{proposition}
\label{prop:strongBig}
Given a simple binary Boolean VCSP instance $\mathcal{C}$ implementing a fitness function $f$,
if removing the constraint $C_{ij}$ changes the corresponding fitness graph, 
then for at least one $k \in \{i,j\}$ there exists some $x \in \{0,1\}^n$ with $x_i = x_j = 1$ such that:
\begin{equation}
c_{ij} \geq f(x) - f(x[k \mapsto 0])  > 0
\quad \text{or} \quad 
c_{ij} \leq f(x) - f(x[k \mapsto 0])  < 0
\label{eq:cij<0}
\end{equation}
\end{proposition}
\begin{proof}
Without loss of generality (by swapping $i$ and $j$ in the variable numbering if necessary), 
we can suppose that $k = i$.
Consider two cases:
\begin{description}
\item[\textit{Case 1}] ($c_{ij} > 0$):
If removing $C_{ij}$  changes the fitness graph, then there exists some 
$x \in \{0,1\}^n$ with $x_i = x_j = 1$ such that:
\begin{equation}
    f(x) > f(x[i \mapsto 0])
    \text{ \; but \;}
    f(x) - c_{ij} \leq f(x[i \mapsto 0]).
    \label{eq:posnotpos}
\end{equation}
We can re-arrange Equation~\ref{eq:posnotpos} to get
$
    c_{ij} \geq f(x) - f(x[i \mapsto 0]) > 0
$.

\item[\textit{Case 2}] ($c_{ij} < 0$): 
This is the same as case 1, except that the direction of the inequalities in Equation~\ref{eq:posnotpos} are reversed.
\end{description}
\vspace{-0.6cm}
\end{proof}
\vspace{-0.6cm}
\begin{proposition}
\label{prop:magSignDep}
Given a simple binary Boolean VCSP instance $\mathcal{C}$ implementing a fitness function $f$,
if there exists a constraint $C_{ij}$ in $\mathcal{C}$, some assignment $x \in \{0,1\}^n$ 
with $x_i = x_j = 1$, and some $k \in \{i,j\}$ such that:
\begin{equation}
c_{ij} \geq f(x) - f(x[k \mapsto 0])  > 0
\quad \text{or} \quad 
c_{ij} \leq f(x) - f(x[k \mapsto 0])  < 0
\label{eq:cij>0Prop}
\end{equation}
then $i$ sign-depends on $j$ in the associated fitness graph $G_\mathcal{C}$.
\end{proposition}
\begin{proof}
As in the proof of Proposition~\ref{prop:strongBig}, 
we can suppose that $k = i$ (by swapping $i$ and $j$ in the variable numbering if necessary).
Also, as in the proof of Proposition~\ref{prop:strongBig}, 
the case for $c_{ij} < 0$ is symmetric 
(by flipping the direction of inequalities) to $c_{ij} > 0$.
Thus, we will just consider the case where $k = i$ and $c_{ij} > 0$.

Given that Equation~\ref{eq:cij>0Prop} tells us that $f(x) > f(x[i \mapsto 0])$
(i.e., that $(x[i \mapsto 0],x) \in E(G_\mathcal{C})$), to establish that $i$ sign-depends on $j$ per Definition~\ref{def:signDep}, we need to show that $f(x[j \mapsto 0]) \leq f(x[i \mapsto 0,j \mapsto 0])$ (i.e., that $(x[i \mapsto 0, j \mapsto 0],x[j \mapsto 0]) \not\in E(G_\mathcal{C})$).
So, let us look at the difference of the latter:
\begin{equation*}
    f(x[j \mapsto 0]) - f(x[i \mapsto 0,j \mapsto 0]) = f(x) - f(x[i \mapsto 0]) - c_{ij}  \leq 0 
\end{equation*}
where the equality follows from Definition~\ref{def:VCSP} ($\mathcal{C}$ implements $f$) 
and Definition~\ref{def:simple} ($\mathcal{C}$ is simple),
and the inequality follows from the first part of Equation~\ref{eq:cij>0Prop}.
\end{proof}
\begin{proof}[of \textbf{Theorem~\ref{thm:trimprep}}]
Note that Equations~\ref{eq:cij<0} and~\ref{eq:cij>0Prop} specify the same conditions,
hence the negation of this condition
can be used to glue together the contrapositives of Proposition~\ref{prop:magSignDep} (if $i$ and $j$ are sign-independent then Equation~\ref{eq:cij<0} does not hold) 
and Proposition~\ref{prop:strongBig} (if Equation~\ref{eq:cij>0Prop} does not hold then $C_{ij}'$ can be removed from $\mathcal{C}'$ without changing the corresponding fitness graph).
So we can convert $\mathcal{C'}$ to a trim VCSP instance 
that is sign-equivalent to $\mathcal{C'}$ by simply removing all $C'_{ij} \in \mathcal{C}'$ where $i$ and $j$ are sign-independent in the associated fitness graph $G_{\mathcal{C}'}$.
\end{proof}
The next result is the sign-equivalence analog of Theorem~\ref{thm:valmin}.
It shows that a trim VCSP instance has the minimal constraint graph of any binary instance with 
the same associated fitness graph. 
\begin{theorem}
\label{thm:signmin}
Let $\mathcal{C}$ be a trim binary Boolean VCSP instance.
If the binary Boolean VCSP instance $\mathcal{C}'$ is sign-equivalent to $\mathcal{C}$, then
$E(\mathcal{C}) \subseteq E(\mathcal{C'})$.
\end{theorem}
To prove Theorem~\ref{thm:signmin}, we just need to show that 
constraints between sign-interacting positions cannot be removed while preserving sign-equivalence.
That is, we just need the following proposition:
\begin{proposition}
\label{prop:signNonZero}
Let $\mathcal{C}$ be a binary Boolean VCSP instance with associated fitness graph $G_\mathcal{C}$.
If $i,j$ sign-interact in $G_\mathcal{C}$,
then the constraint $C_{ij}$ in $\mathcal{C}$ is non-zero.
\end{proposition}
\begin{proof}
Without loss of generality, assume that $i<j$ and we have an edge in $G_\mathcal{C}$ 
from $x[i \mapsto \bar{x_i}]$ to $x$.
Thus, the fitness function $f$ implemented by $\mathcal{C}$ must satisfy the following two inequalities:
\begin{equation}
f(x)  > f(x[i \mapsto \bar{x_i}]) 
\quad \text{and} \quad
f(x[j \mapsto \bar{x_j}]) \leq f(x[i \mapsto \bar{x_i}, j \mapsto \bar{x_j}]) \label{eq:edges}
\end{equation}
Define $g_i(x) = C_{i}(x_i) + \sum_{k \in  N_{\mathcal{C}}(i)\setminus \{j\}} C_{ik}(x_i,x_k)$ 
and similarly for $g_j$. 
Also let $K_{ij}(x)$ be the part of $f$ independent of $x_i,x_j$: 
i.e., $f(x) = K_{ij}(x) + g_i(x) + g_j(x) + C_{ij}(x_i,x_j)$. 
Rewriting (and simplifying) the two parts of Equation~\ref{eq:edges}, we get:
\begin{align*}
g_i(x) + C_{ij}(x_i,x_j) & > g_i(x[i \mapsto \bar{x_i}]) +  C_{ij}(\bar{x_i},x_j) \\ g_i(x) + C_{ij}(x_i,\bar{x_j}) & \leq 
g_i(x[i \mapsto \bar{x_i}]) +  C_{ij}(\bar{x_i},\bar{x_j}) 
\end{align*}
These equations can be rotated to sandwich the $g_i$ terms:
\begin{equation*}
    C_{ij}(x_i,x_j) - C_{ij}(\bar{x_i},x_j) > g_i(x[i \mapsto \bar{x_i}]) - g_i(x) 
    \geq C_{ij}(x_i,\bar{x_j}) -  C_{ij}(\bar{x_i},\bar{x_j})
\end{equation*}
which simplifies to 
$C_{ij}(x_i,x_j) - C_{ij}(\bar{x_i},x_j) > C_{ij}(x_i,\bar{x_j}) -  C_{ij}(\bar{x_i},\bar{x_j})$ 
and -- due to the strict inequality -- establishes that $C_{ij}$ is non-zero.
\end{proof}
Our next result shows that the signs of the constraint weights 
on this minimal constraint graph are preserved
across all simple VCSPs with the same fitness graph.
This gives another motivation for calling these representations sign-equivalent.
\begin{proposition}
\label{prop:samesign}
Let $\mathcal{C}$ be a simple trim binary Boolean VCSP instance.
If the simple binary Boolean VCSP instance $\mathcal{C}'$ is sign-equivalent to $\mathcal{C}$, then
\begin{enumerate}
\item for all $i \in [n]$ we have that $\sgn(c_{i}) = \sgn(c'_{i})$; and
\item for all $\{i,j\} \in E(\mathcal{C})$ we have that $\sgn(c_{ij}) = \sgn(c'_{ij})$.
\end{enumerate}
\end{proposition}
\begin{proof}
For (1), let $e_i \in \{0,1\}^n$ be the variable assignment that sets the $i$th variable to one, and all other variables to zero, 
and
let $f_\mathcal{C}$ be the fitness function implemented by $\mathcal{C}$.
Note that $c_i = f_\mathcal{C}(e_i) - f_\mathcal{C}(0^n)$ and $c'_i = f_{\mathcal{C}'}(e_i) - f_{\mathcal{C}'}(0^n)$.
Since $\mathcal{C}$ and $\mathcal{C}'$ are sign-equivalent, $\sgn(f_\mathcal{C}(e_i) - f_\mathcal{C}(0^n)) = \sgn(f_{\mathcal{C}'}(e_i) - f_{\mathcal{C}'}(0^n))$.
Hence, $\sgn(c_i) = \sgn(c'_i)$.

For (2), assume for contradiction that $c_{ij} = a$ and $c'_{ij} = -b$ for some positive integer $a$ and non-negative integer $b$ 
(repeat this argument with $-a$ and $b$ for the other case).
Note that the sum of two sign-equivalent VCSP instances is sign-equivalent to both, so $C'' = b\mathcal{C} + a\mathcal{C}'$ is sign-equivalent to $\mathcal{C}$.
But $c''_{ij} = ba - ab = 0$ so $\{i,j\} \not\in E(\mathcal{C}'')$.
Since $\mathcal{C}$ is trim, this contradicts Theorem~\ref{thm:signmin}.
Hence, $\sgn(c_{ij}) = \sgn(c'_{ij})$.
\end{proof}
However, unlike with magnitude-equivalence, it is NP-hard to determine a minimal sign-equivalent VCSP instance, as the next result shows:
\begin{theorem}
\label{thm:weakiscoNPcomplete}
Let $\mathcal{C}$ be a simple binary Boolean VCSP instance 
with associated fitness graph $G_\mathcal{C}$.
The problem of deciding whether $i,j$ sign-interact in $G_\mathcal{C}$ is NP-complete.
\end{theorem}

\begin{proof}
To show that this problem is in NP, we observe that
we can provide a variable assignment $x$ as a certificate and check that under that variable assignment either $i$ sign-depends on $j$ or $j$ sign-depends on $i$ (or both).

We will establish NP-hardness by reduction from the \textsc{SubsetSum} problem, which is known to be NP-complete~\cite{Garey1979}: A set of integers $\{s_1,\ldots,s_n\}$ and a target $t$ is a yes-instance of the \textsc{SubsetSum} problem if there exists some subset $S \subseteq [n]$ such that $\sum_{i \in S}s_i = t$.

Now consider a simple binary Boolean VCSP instance $\mathcal{C}$ on $n + 2$ variables,
that implements fitness function $f$ and has associated fitness graph $G_\mathcal{C}$,
whose constraint graph has the shape of a star, with variable ${n+2}$ at the centre
(see Figure~\ref{fig:subsetsumgadget}).

\begin{figure}[hbt]
\begin{center}
\includegraphics[width=0.77\textwidth]{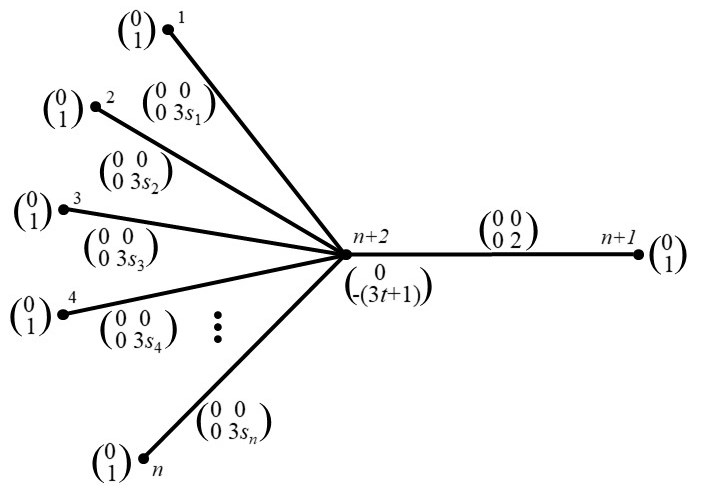}
\end{center}
\vspace{-0.3cm}
\caption{Binary VCSP instance used in the proof of Theorem~\protect\ref{thm:weakiscoNPcomplete}}
\label{fig:subsetsumgadget}
\end{figure}

\noindent The constraints of $\mathcal{C}$ are given by:
\begin{itemize}
\item unary constraints $C_i = \uC{0}{1}$ for all $i \leq n + 1$  and $C_{n + 2} = \uC{0}{-(3t + 1)}$; 
\item binary constraints $C_{i,n+2} = \bC{0}{0}{0}{3s_i}$ between the central variable ${n + 2}$ and variable $i$, for $1 \leq i \leq n$; and 
\item binary constraint $C_{n + 1, n + 2} = \bC{0}{0}{0}{2}$ between variables ${n+1}$ and ${n+2}$.
\end{itemize}

\noindent \textit{Claim:} $\langle \{s_1,\ldots,s_n\},t \rangle$
is a yes-instance of \textsc{SubsetSum} if and only if 
${n+1}$ and ${n+2}$ sign-interact.

We clearly have that for all $x \in \{0,1\}^{n + 2}$, $f(x[{n + 1} \mapsto 1]) > f(x[{n + 1} \mapsto 0])$, so ${n + 1}$ does not sign-depend on ${n + 2}$.
Thus our claim becomes equivalent to verifying the conditions under which 
${n + 2}$ sign-depends on ${n + 1}$.
Let's look at the two directions of the if and only if in the claim:
\begin{description}
\item[\textit{Case 1} ($\Rightarrow$):]
If $\langle \{s_1,\ldots,s_n\},t \rangle \in \textsc{SubsetSum}$, 
then there is a subset $S \subseteq [n]$ such that $\sum_{i \in S}s_i = t$.
Let $e_S \in \{0,1\}^n$ be the variable assignment such that for any $i \in S$, $e_S[i] = 1$ 
and for any $j \not\in S$, $e_S[j] = 0$.
We have that:
\begin{align*}
    f(e_S01) &=  |S| - 1 \quad     & f(e_S11) &=  |S| + 2\\  
    f(e_S00) &=  |S|     \quad & f(e_S10) &=  |S| + 1
\end{align*}
By Equation~\ref{eq:signdep2}, these imply that ${n + 2}$ sign-depends on ${n + 1}$.
\item[\textit{Case 2} ($\Leftarrow$):]
If $\langle \{s_1,\ldots,s_n\},t \rangle \not\in \textsc{SubsetSum}$, then
for any $S \subseteq [n]$ we either have $\sum_{i \in S}s_i \leq t - 1$ or $\sum_{i \in S}s_i \geq t + 1$. 
Thus, given an arbitrary assignment $e_S \in \{0,1\}$:
\begin{align*}
    &\text{If }\sum_{i \in S}s_i \leq t - 1 \text{ then:}\quad\quad
    && \text{Or, if } \sum_{i \in S}s_i \geq t + 1 \text{ then:} \\
    &\quad f(e_S01) - f(e_S00) \leq -4 &&\quad f(e_S01) - f(e_S00) \geq 2\\
    &\quad f(e_S11) - f(e_S10) \leq -2 &&\quad f(e_S11) - f(e_S10) \geq 4
\end{align*}
In either subcase, $\sgn(f(e_S01) - f(e_S00)) = \sgn(f(e_S11) - f(e_S10))$, 
so by Equation~\ref{eq:signdep2}, ${n + 2}$ does not sign-depend on ${n + 1}$.
\end{description}
\vspace{-0.5cm}
\end{proof}

\section{Span}
\label{sec:span}

In this section we show that a simple function of the numerical 
values of the constraints in a VCSP instance 
provides an upper bound on the length of the longest directed path 
in the associated fitness graph,
and hence a bound on the number of steps taken by any local search algorithm.

\begin{definition}
\label{def:span}
Given a VCSP instance $\mathcal{C}$ over domain $D^{[n]}$, define \[\Span(\mathcal{C}) = 
\sum_{C_S \in \mathcal{C}} \left( \max_{z \in D^{S}} C_{S}(z) - \min_{z \in D^{S}} C_{S}(z) \right).
\]
\end{definition}

\begin{proposition}
\label{prop:spanPath}
Given any VCSP instance $\mathcal{C}$,  
the length of the longest directed path in the associated fitness graph $G_{\mathcal{C}}$ is less than or equal to $\Span(\mathcal{C})$.
\end{proposition}
\begin{proof}
The maximum value of the fitness function $f$ implemented by $\mathcal{C}$ cannot exceed the sum of the largest fitness values assigned by each constraint.
Similarly, the minimum value of $f$ cannot be less than the sum of the smallest fitness values assigned by each constraint.
The difference between these bounds is precisely $\Span(\mathcal{C})$.
Since we have defined a VCSP instance, and hence the associated fitness function $f$, to be integer-valued, each directed edge in the fitness graph increases fitness by at least one, so there can be at most $\Span(\mathcal{C})$ many such steps in any path.
\end{proof}
Although the span of a VCSP instance always provides an upper bound on 
the length of the longest ascent in the associated fitness landscape,
in general this bound is not tight, as the following simple example shows: 
\begin{example}
\label{ex:smoothSpan}
Consider the unary VCSP instance $\mathcal{C} = \{ C_{i} = \begin{pmatrix} 0 \\ 2^i \end{pmatrix} |\; i = 1\ldots n \}$.

\noindent The longest directed path in the associated fitness graph $G_{\mathcal{C}}$ 
is of length $n$, but $\Span(\mathcal{C}) = 2^{n + 1} - 1$.
\end{example}
To avoid such a large discrepancy between the longest 
path in the fitness graph and the $\Span$,
and hence obtain a tighter bound, we can look for sign-equivalent VCSP instances that have a minimal span.
\begin{example}
\label{ex:smoothSpanMinimised}
In the case of Example~\ref{ex:smoothSpan}, a sign-equivalent minimal-span VCSP instance is given by $\mathcal{C}' = \{ C_{i} = \begin{pmatrix} 0 \\ 1 \end{pmatrix} |\; i = 1\ldots n\}$, which has $\Span(\mathcal{C}') = n$ -- giving us a tight bound on the length of the longest improving path.
\end{example}

Finding a sign-equivalent instance
with the smallest possible span may not always be straightforward.
However, for simple binary Boolean VCSP instances the span is just the sum of the
absolute values of the constraint weights.
Hence, for any binary Boolean VCSP instance we can compute the minimal span of a 
sign-equivalent simple binary Boolean instance by solving 
an integer linear program over the constraint weights.

Before describing this linear program, we show that restricting the search for sign-equivalent instances to
simple trim instances only changes the minimal span that can be obtained by a 
small constant factor.
\begin{theorem}
\label{thm:spanEq}
For any binary Boolean VCSP instance $\mathcal{C}'$,
there exists a sign-equivalent simple trim binary Boolean VCSP instance $\mathcal{C}$ such that 
$\Span(\mathcal{C}) \leq 4\;\Span(\mathcal{C}')$.
\end{theorem}
We establish Theorem~\ref{thm:spanEq} by showing that if we start with
an arbitrary binary Boolean VCSP instance $\mathcal{C}'$ 
and transform it to a simple, trim, sign-equivalent
VCSP $\mathcal{C}$, as described in the previous sections, 
then that will increase the span by at most a factor of $4$.
This proceeds by two steps:
\begin{proposition}
\label{prop:simpleSpan}
Any binary Boolean VCSP instance $\mathcal{C}'$ can be transformed into
a simple VCSP instance $\mathcal{C}$ that is magnitude-equivalent to 
$\mathcal{C}'$ with 
$\mathrm{span}(\mathcal{C}) \leq 4\; \mathrm{span}(\mathcal{C}')$.
\end{proposition}
\begin{proof}
Let us decompose the span of $\mathcal{C}$ 
into the contribution due to unary constraints ($\Span_1$) and binary constraints ($\Span_2$):
\begin{equation*}
    \Span(\mathcal{C}) = \overbrace{\sum_{i\in[n]}|c_i|}^{\Span_1(\mathcal{C})} + 
    \overbrace{\sum_{\{i,j\} \in E(\mathcal{C})}|c_{ij}|}^{\Span_2(\mathcal{C})}
\end{equation*}
Using Equation~\ref{eq:simple_ci} from the proof of Theorem~\ref{thm:simplerep}, we can express $|c_i|$ in terms of $\mathcal{C}'$ as:
\begin{align}
    |c_i| & = |\;C'_{i}(1) - C'_{i}(0) \quad + \sum_{j \; | \; \{i,j\} \in E(\mathcal{C'})}C'_{ij}(1,0) - C'_{ij}(0,0])\;| \nonumber \\
    & \leq |\;C'_{i}(1) - C'_{i}(0)\;| \quad + \sum_{j \; | \; \{i,j\} \in E(\mathcal{C}')} |\; C'_{ij}(1,0) - C'_{ij}(0,0)\;| \nonumber \\
     & \leq \Span(\{C'_i\}) + \Span(\{C'_{ij} \; | \; j \in N_{\mathcal{C}'}(i) \}) \label{eq:2spans}
\end{align}
Notice that the first span in Equation~\ref{eq:2spans} is of the unary constraint in $\mathcal{C}'$ that has $i$ as its scope, and the second span is of all binary constraints in $\mathcal{C}'$ that have $i$ in scope (or, equivalently: span of all edges incident on $i$ in the constraint graph of $\mathcal{C}'$).
This means that if we sum $|c_i|$ over all $i \in [n]$ then we cover the whole graph:
\begin{align}
    {\textstyle \Span_1}(\mathcal{C}) = \sum_{i = 1}^n |c_i| & \leq \sum_{i = 1}^n \Span(\{C'_i\}) + \sum_{i = 1}^n \Span(\{C'_{ij} \; | \; j \in N_{\mathcal{C}'}(i) \}) \nonumber \\
    & = {\textstyle \Span_1}(\mathcal{C}') + 2\;{\textstyle \Span_2}(\mathcal{C}') \label{eq:doublecover}\\
    & \leq 2\;\Span(\mathcal{C}') \label{eq:span1C}
\end{align}
where Equation~\ref{eq:doublecover} has a double cover in its second summand because each edge in the constraint graph of $\mathcal{C}'$ has two end points (equivalently: all scopes are binary).

Similarly, using Equation~\ref{eq:simple_cij} from the proof of Theorem~\ref{thm:simplerep}, we can also express $|c_{ij}|$ in terms of $\mathcal{C}'$ as:
\begin{align*}
    |c_{ij}| & = |C\;'_{ij}(0,0) - C'_{ij}(0,1) - C'_{ij}(1,0) + C'_{ij}(1,1)\;| \\
    & \leq  |\;C'_{ij}(1,1) - C'_{\{i,j\}}(1,0)\;| + |\;C'_{ij}(0,0) - C'_{ij}(0,1)\;| \\
    & \leq 2\;\Span(\{C'_{ij}\}) 
\end{align*}
As before, if we sum $|c_{ij}|$ over all $\{i,j\} \in E(\mathcal{C})$ then we cover the whole graph:
\begin{equation}
    {\textstyle \Span_2}(\mathcal{C}) = \sum_{\{i,j\} \in E(\mathcal{C})}|c_{ij}| 
    \leq 2\sum_{\{i,j\} \in E(\mathcal{C}')}\Span(\{C'_{ij}\}) 
    \leq 2\;{\textstyle \Span_2}(\mathcal{C}') \label{eq:span2C}
\end{equation}
where for the last inequality we moved from summing over $\{i,j\} \in E(\mathcal{C})$ to $\{i,j\} \in E(\mathcal{C}')$ because $ E(\mathcal{C}) \subseteq E(\mathcal{C}')$ by Theorem~\ref{thm:valmin}.
Combining Equations~\ref{eq:span1C} and \ref{eq:span2C}, we get the final result that $\Span(\mathcal{C}) = \Span_1(\mathcal{C}) + \Span_2(\mathcal{C}) \leq 4\;\Span(\mathcal{C}')$.
\end{proof}
Note that Proposition~\ref{prop:simpleSpan} is the best possible, 
as the following example shows:
\begin{example}
The two VCSP instances shown in Figure~\ref{fig:magequiv2} are magnitude-equivalent.
\begin{figure}[htb]
\begin{center}
\includegraphics[width = 0.5\textwidth]{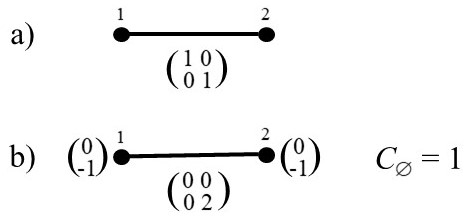}
\end{center}
\caption{VCSP instances with differing span that implement the same fitness function} 
\label{fig:magequiv2}
\end{figure}
\noindent The instance in Figure~\ref{fig:magequiv2}a has a span of $1$ and the simple instance in Figure~\ref{fig:magequiv2}b has a span of $1 + 2 + 1 = 4$.
Thus, sometimes the simplifying procedure from Section~\ref{sec:mageq} can increase span by a factor of $4$.
\end{example}

In contrast, the trimming procedure from Section~\ref{sec:signeq} can only decrease span:
\begin{proposition}
For any simple binary Boolean VCSP instance $\mathcal{C}'$,
there exists a sign-equivalent simple trim binary Boolean VCSP instance $\mathcal{C}$ such that $\Span(\mathcal{C}) \leq \Span(\mathcal{C}')$
\label{prop:trimSpan}
\end{proposition}
\begin{proof}
The trimming procedure from Section~\ref{sec:signeq} removes constraints but doesn't change any remaining ones.
\end{proof}
Combining Propositions~\ref{prop:simpleSpan} and \ref{prop:trimSpan}, 
establishes Theorem~\ref{thm:spanEq}.

Theorem~\ref{thm:spanEq} implies that, in the binary Boolean case,
restricting our search for a minimal-span sign-equivalent instance 
to simple trim instances will only increase the span obtained 
by a factor of at most four.
Hence, given any binary Boolean VCSP instance $\mathcal{C}$, 
we will formulate the problem of finding a minimal-span
sign-equivalent instance as an optimisation problem 
over the weights of a simple binary Boolean instance on 
the same constraint graph. 
This will considerably simplify the search procedure, 
compared with searching for an arbitrary minimal-span sign-equivalent instance.

Given a simple trim binary Boolean VCSP instance, $\mathcal{C}$, 
it follows from Proposition~\ref{prop:samesign} that 
any sign-equivalent simple instance must 
preserve the signs of each of the constraint weights.
Moreover, the values of these weights must satisfy a collection of linear inequalities,
to ensure that the edges of the fitness graph are preserved,
as indicated in Equation~\ref{eq:cij<0}.
Each of these inequalities constrains the weights of 
constraints with scopes including a particular variable.
We therefore have the optimisation problem given in Definition~\ref{def:liftedCSP},
which involves linear inequality and equality constraints of the following forms:
\begin{definition}
\label{def:metaC}
Given two sets of variables $L$ and $R$, let $\leq_{+k}[L,R]$ be a constraint of arity $|L| + |R|$ that is satisfied when
$k + \sum_{x \in L} x \leq \sum_{y \in R} y$.
Call $L$ the \emph{left side} of the constraint and $R$ the \emph{right side}.
Similarly, define $=_{+k}[L,R]$ as above but with $\leq$ replaced by $=$.
\end{definition}

Using these two constraint types, we can now define the 
span-minimisation problem for binary Boolean VCSP instances as follows:
\begin{definition}
\label{def:liftedCSP}
Given a simple binary Boolean VCSP instance $\mathcal{C}$ on $n$ 
variables (with a constraint graph that has neighbourhood function $N: [n] \rightarrow 2^{[n]}$), the corresponding 
\textbf{span-minimisation problem} for $\mathcal{C}$ has  
a set of variables $V$, each with domain $\mathbb{N}$, where:
\[
V = \{p_i \mid C_{i} \in \mathcal{C}\} \cup 
\{p_{\{i,j\}} \mid C_{ij} \in \mathcal{C}\}.
\]
Divide the set $V$ into two sets $V_+,V_-$ with $p_i \in V_+$ or $p_{\{i,j\}} \in V_+$ if $c_i > 0$ or $c_{ij} > 0$ and otherwise $p_i \in V_-$ or $p_{\{i,j\}} \in V_-$ if $c_i < 0$ or $c_{ij} < 0$.

For each variable $i \in [n]$ of $\mathcal{C}$, introduce $2^{|N(i)|}$  linear constraints, one for each $Y \subseteq \{p_{\{i,j\}} \mid j \in N(i)\}$, 
depending on the sign of $s = c_i + \sum_{j \in Y}c_{ij}$:
\begin{itemize}
    \item If $s < 0$ then add the constraint $\leq_{+1}[(Y\cup\{p_i\})\cap V_+,(Y\cup\{p_i\})\cap V_-]$, else
    \item If $s = 0$ then add the constraint $=_{+0}[(Y\cup\{p_i\})\cap V_+,(Y\cup\{p_i\})\cap V_-]$, else
    \item If $s > 0$ then add the constraint $\leq_{+1}[(Y\cup\{p_i\})\cap V_-,(Y\cup\{p_i\})\cap V_+]$.
\end{itemize}
A \emph{feasible solution} to this span-minimisation problem is 
an assignment of values to the variables in $V$ that satisfies all the constraints. A feasible solution is \emph{optimal} if it
minimises the sum of the assigned values.
\end{definition}
Note that there is at least one feasible solution to 
the span-minimisation problem for $\mathcal{C}$  
given by $p_i = |c_i|$ and $p_{\{i,j\}} = |c_{ij}|$ (i.e., the absolute values of the weights of the original simple VCSP constraints).
\begin{example}{\bf (Unary instances have linear minimal span)}
For any \emph{unary} Boolean VCSP instance with $n$ variables, the span-minimisation problem 
only has inequalities of the form $1 \leq p_i$.
Hence an optimal 
solution sets each $p_i$ to $1$, so the corresponding 
minimal-span sign-equivalent instance has 
each $c_i$ equal to $1$ or $-1$, and so has span $n$. 
This is precisely the length of the longest improving path, as illustrated in Example~\ref{ex:smoothSpanMinimised}.
\end{example}
\begin{example}{\bf (Degree 2 binary instances have quadratic minimal span)}
\label{ex:cycle}
Consider the binary Boolean VCSP instance on $8$ variables
shown in Figure~\ref{fig:CycleSpanExample}, where the constraint graph is 
a cycle.
The span of this instance is quite large, but if we convert to a simple instance, and then solve the span-minimisation problem, we obtain a sign-equivalent instance with a much smaller span,
as shown in Figure~\ref{fig:CycleSpanExampleMinimised}.

\begin{figure}[htb]
    \centering
    \begin{subfigure}[b]{0.49\textwidth}
    \includegraphics[width=\textwidth]{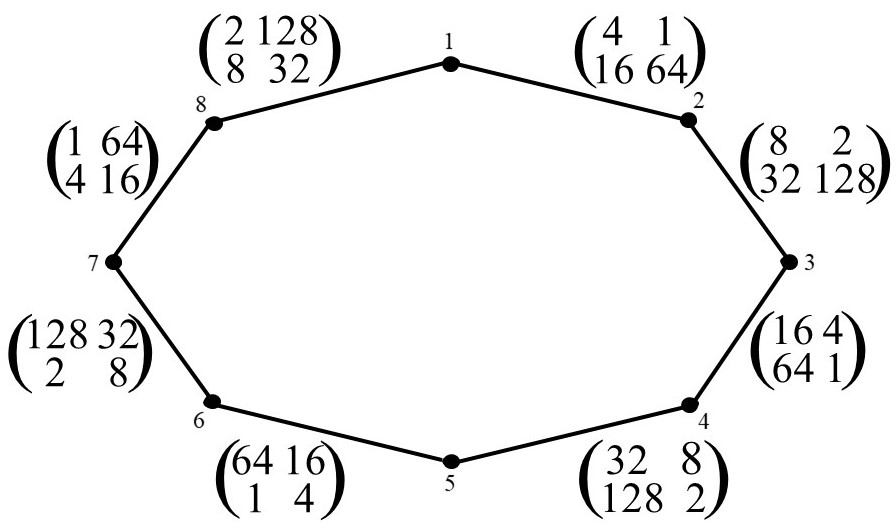}
    \caption{Original with span of 756}
    \label{fig:CycleSpanExample}
    \end{subfigure} \hfill
    \begin{subfigure}[b]{0.49\textwidth}
    \centering
    \includegraphics[width=0.92\textwidth]{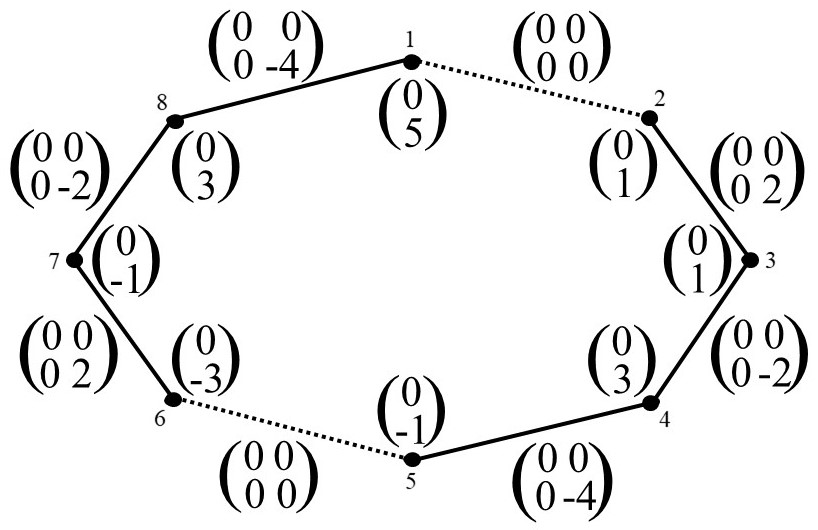}
    \caption{Minimised span of 34}
    \label{fig:CycleSpanExampleMinimised}
    \end{subfigure}
    \caption{Two sign-equivalent binary Boolean VCSP instances of different span}
\end{figure}

In general, by analysing the structure of the inequalities in Definition~\ref{def:liftedCSP},
it can be shown that for any binary Boolean VCSP instance whose constraint graph has maximum degree 2
(e.g., a cycle) 
the weights of a minimal-span sign-equivalent simple instance can grow only linearly with the number of variables, and
hence the span and longest improving path in the associated fitness graph can grow only quadratically at most.
For a full analysis, see Section 5.1.1 of \citeauthor{KazThesis}~\citeyear{KazThesis}.
\end{example}

Unfortunately, span arguments like the one above cannot be used to obtain
tight bounds for all binary Boolean VCSP instances,
as we show with the following example of a VCSP with only short 
improving paths but exponential minimal span.

\begin{example}{\bf (Large span in tree-structured constraint graph)}
\label{ex:bigSpanTree}
Consider the family of binary Boolean VCSP instances on $3K+2$ variables 
illustrated in Figure~\ref{fig:bigSpantree},
where the constraint graph of each instance is a tree.
\begin{figure}[phtb]
\begin{center}
\vspace{0.3cm}
\includegraphics[width=0.9\textwidth,height=4.1cm]{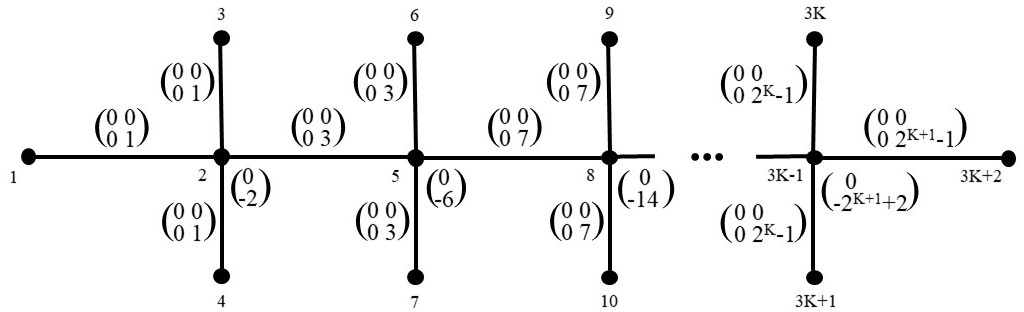}
\end{center}
\vspace{-0.4cm}
\caption{A family of tree-structured simple binary Boolean VCSP instances}
\label{fig:bigSpantree}
\end{figure}
Solving the span-minimisation problem shows that the weight values cannot be reduced any
further without changing the fitness graph, hence all such instances are span-minimal. 

The span of each such instance is $3(2^{K+2}-5K-9)$, 
and hence grows exponentially with the number of variables.
However, the longest improving path in the associated fitness graph grows only quadratically with the number of variables (Theorem~\ref{thm:main}).
\end{example}

In the next section, we will show that
for any binary Boolean VCSP instance whose constraint graph is a tree (like Example~\ref{ex:bigSpanTree}),
the longest improving path in the associated fitness graph 
is bounded by a quadratic function of the number of variables.
This analysis will require us to develop more sophisticated techniques than simply computing the span.

\section{Tree-Structured Boolean VCSP Instances}
\label{sec:tree-structured}

In this section, we will prove the following:
\begin{theorem}
\label{thm:main}
For any binary Boolean VCSP instance $\mathcal{C}$ on $n$ variables, if the constraint-graph of $\mathcal{C}$ is a tree, then any directed path in the associated fitness graph $G_\mathcal{C}$ 
has length at most ${n \choose 2} + n$.
\end{theorem}
Note that this result bounds the length of \emph{any} directed path in $G_\mathcal{C}$, 
not just the path taken by a particular local-search algorithm.
Thus, on such landscapes even choosing the worst possible sequence of improving moves 
results in a local optimum being found in polynomial time.

We will show in Section~\ref{sec:robustness} that 
being Boolean and tree-structured are both essential conditions 
to obtain a polynomial bound on the length of all improving paths
in the associated fitness graph.

To see that the bound in Theorem~\ref{thm:main} is the best possible for binary Boolean tree-structured VCSP instances, consider the path-structured VCSP instance on $n$ variables 
described in Example~\ref{ex:quadpath}.
\begin{example}{\bf (Path of length ${n \choose 2} + n$)}
\label{ex:quadpath}
Consider the binary Boolean VCSP instance on $n$ variables
illustrated in Figure~\ref{fig:quadpath},
where the constraint graph is a path, and the constraint on each edge $\{i,i+1\}$
is $\bC{i}{0}{0}{i}$.
\begin{figure}[htb]
\begin{center}
\includegraphics[width=0.9\textwidth]{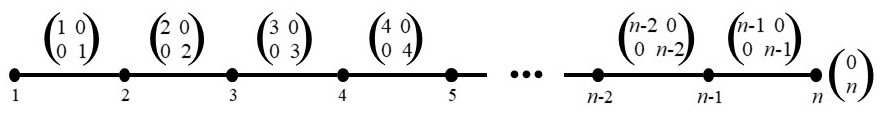}
\end{center}
\vspace{-0.5cm}
\caption{A binary Boolean VCSP instance with an improving path of length ${n \choose 2} + n$ 
in the associated fitness graph}
\label{fig:quadpath}
\end{figure}

To obtain an improving path of length ${n \choose 2} + n$ in the corresponding fitness graph, 
consider an initial variable assignment of $x = (10)^\frac{n}{2}$ if $n$ is even and $x = (10)^\frac{n - 1}{2}1$ if $n$ is odd, and always select the leftmost variable that is able to flip.
This will increase the fitness by $1$ at each step, starting from $0$ to $\frac{n(n + 1)}{2}$.

For example, when $n=4$, this gives the following sequence of eleven assignments, each of which 
increases the value of the fitness function by one: 

\begin{equation}
\small{1010 \!\rightarrow\! 0010 \!\rightarrow\! 0110 \!\rightarrow\! 1110 \!\rightarrow\! 1100 \!\rightarrow\! 1000
\!\rightarrow\! 0000 \!\rightarrow\! 0001 \!\rightarrow\! 0011 \!\rightarrow\! 0111 \!\rightarrow\! 1111}
\label{eq:expath}
\end{equation}
\end{example}

For the proof of Theorem~\ref{thm:main}, we introduce some further definitions.
\begin{definition}
\label{def:flip}
Given any directed path $p = x^1\ldots x^t\ldots x^T$ in a 
fitness graph $G$, where each $x^t$ is a Boolean assignment, define the \textbf{flip function} $m$ at time $t$ as follows: $m(t) = (i \mapsto b)$ 
where $x^{t + 1} \oplus x^t = e_i$ and $b = x^{t + 1}_i$ 
(i.e., the $i$-th variable is flipped at time $t$ to value $b$).
\end{definition}
For illustration, consider the sequence of moves listed in Equation~\ref{eq:expath} of Example~\ref{ex:quadpath}.
It corresponds to the following flip function:
\begin{center}
\begin{tabular}{c||c|c|c|c|c|c|c|c|c|c|}
     \textbf{t} & $1$ & $2$ & $3$ & $4$ & $5$ & $6$ & $7$ & $8$ & $9$ & $10$\\
     \hline
     \textbf{m(t)} & $1 \mapsto 0$ & $2 \mapsto 1$ & $1 \mapsto 1$ & $3 \mapsto 0$ & $2 \mapsto 0$ & $1 \mapsto 0$ & $4 \mapsto 1$ & $3 \mapsto 1$ & $2 \mapsto 1$ & $1 \mapsto 1$
\end{tabular}
\end{center}
To obtain the bound on the length of all paths in a fitness graph given by Theorem~\ref{thm:main}, we will identify conditions on any possible flip function, and hence bound the maximum possible value for $T$.
\begin{definition}
\label{def:gain}
Given any fitness function $f$ on Boolean assignments, 
define the \textbf{gain function} for value $b$ in position $i$ of
an assignment $x$ as follows:
$$\gain(x,i,b) = f(x[i \mapsto b]) - f(x[i \mapsto \bar{b}]).$$
\end{definition}
\begin{definition}
\label{def:support}
We say that a flip $m(t') = (j \mapsto c)$ \emph{\bf supports} a flip $m(t) = (i \mapsto b)$ 
if 
\begin{equation}
    \gain(x^{t'}[j \mapsto c],i,b) > 0 \geq \gain(x^{t'}[j \mapsto \bar{c}],i,b)
    \label{eq:supportSign}
\end{equation}
\noindent and $t' < t$.
If $x^t_j = c$, then the support is said to be \emph{\bf strong}.
\end{definition}
It is useful to note that the inequality in Equation~\ref{eq:supportSign} implies that position $i$ sign-depends on position $j$ and also means that $C_{ij}(b,c) - C_{ij}(\bar{b},c) > C_{ij}(b,\bar{c}) - C_{ij}(\bar{b},\bar{c})$.
This inequality on $C_{ij}$ is symmetric in the sense that:
\begin{equation}
\begin{aligned}
\quad C_{ij}(b,c) - C_{ij}(\bar{b},c) > C_{ij}(b,\bar{c}) - C_{ij}(\bar{b},\bar{c}) \\ \Leftrightarrow
C_{ij}(b,c) - C_{ij}(b,\bar{c}) > C_{ij}(\bar{b},c) - C_{ij}(\bar{b},\bar{c})
\\ \Leftrightarrow
C_{ji}(c,b) - C_{ji}(\bar{c},b) > C_{ji}(c,\bar{b}) - C_{ji}(\bar{c},\bar{b})
\end{aligned} \label{eq:swapij}
\end{equation}
\begin{definition}
\label{def:encourage}
We say that a flip $m(t') = (j \mapsto c)$ \emph{\bf encourages} a flip $m(t) = (i \mapsto b)$,
and write $(t',j \mapsto c) \encourages (t,i \mapsto b)$,
if $m(t')$ is the most recent flip that strongly supports $m(t)$.
If there are no flips that encourage $m(t)$, then we say that 
$m(t)$ is \emph{\textbf{courageous}}, 
and write $\bot \encourages (t,i \mapsto b)$.
\label{def:cour}
\label{def:enc}
\end{definition}
For illustration, consider the sequence of moves listed in Equation~\ref{eq:expath} of Example~\ref{ex:quadpath}.
It corresponds to the following encouragement relation:
\begin{gather*}
    \bot \encourages (1,1 \mapsto 0), \\
    \bot \encourages (2,2\mapsto 1)  \encourages (3,1\mapsto 1), \\
    \bot \encourages (4,3 \mapsto 0)  \encourages (5,2 \mapsto 0)  \encourages (6, 1 \mapsto 0), \\
    \bot \encourages (7,4 \mapsto 1)  \encourages (8,3 \mapsto 1)  \encourages (9,2 \mapsto 1)  \encourages (10, 1 \mapsto 1) \\
\end{gather*}
Note that for any path $p = x^1 \ldots x^t \ldots x^T$ in a fitness graph,
the value of the associated fitness function $f$ must strictly increase at each step, so 
if $m(t) = (i \mapsto b)$ then $f(x^t[i \mapsto b]) > f(x^t[i \mapsto \bar{b}])$.
\begin{proposition}
If $(t_1,j \mapsto c) \encourages (t_2,i \mapsto b)$ (or if $\bot \encourages (t_2,i \mapsto b)$, set $t_1 = 0$)  
then for all $t$ with $t_1 < t \leq t_2$ we have $\gain(x^t,i,b) > 0$.
\label{prop:enc2fi}
\end{proposition}
\begin{proof}
Since $m(t_2) = (i \mapsto b)$, we have that $\gain(x^{t_2},i,b) > 0$.

\noindent Now assume, for contradiction, that for some $t_0$ with 
$t_1 < t_0 < t_2$ we have $\gain(x^{t_0},i,b) \leq 0$.

In order to change the assignment $x^{t_0}$ so that $\gain(x^{t_2},i,b) > 0$, there
must be at least one flip $m(t')$ for some value of $t'$ with $t_0 < t' < t_2$
which increases the value of the gain function at $x^{t_2}$ above zero, and hence 
strongly supports $m(t_2)$.

Hence $m(t_1)$ is not the most recent flip that satisfies the conditions in 
Definition~\ref{def:encourage}, which contradicts the assertion that 
$(t_1,j \mapsto c) \encourages (t_2,i \mapsto b)$.
\end{proof}
By Definition~\ref{def:enc}, each flip can only be encouraged by at most one other flip, that occurs earlier in time,
so each node in the graph of the encouragement relation has out-degree at most one, and is acyclic.
Directed acyclic graphs where each vertex has at most one parent are forests, so the encouragement graph is a forest. 
This forest has a component for each courageous flip, and we will now show that
there are at most $n$ of these:
\begin{proposition}
\label{prop:CourStart}
At each variable position $i$, only the first flip can be courageous.
\end{proposition}
\begin{proof}
Consider a courageous flip $\bot \encourages (t,i \mapsto b)$, by Proposition~\ref{prop:enc2fi}, we know that for all $t' < t$,
$\gain(x^t,i,b) > 0$.
Thus, there is no time $t' \leq t$ such that $i$ could have flipped to $\bar{b}$: 
hence $i$ was always at $\bar{b}$ for $t' \leq t$. 
So the courageous flip had to be the first flip at that position.
\end{proof}
We will now prove that an encouragement tree cannot double-back on itself in position (Proposition~\ref{prop:noReverse}), and that every branch is a branch in position (Proposition~\ref{prop:noTime}).
When the constraint graph is itself a tree, 
this will imply that each tree in the encouragement forest is a sub-tree of the constraint graph.
\begin{proposition}
\label{prop:noReverse}
If $(t_1,i \mapsto a) \encourages (t_2,j \mapsto b) \encourages (t_3,k \mapsto c)$, then $i \neq k$
\end{proposition}
\begin{proof}
Since $(t_1,i \mapsto a)$ encourages $(t_2,j \mapsto b)$, we have $x^{t_2}_i = a$. 
If, for the sake of contradiction, we assume that $i = k$ then $a = c$ (because if we had $c = \bar{a}$ then the two encouragements would force a contradiction via clashing Equations~\ref{eq:swapij})
and by Proposition~\ref{prop:enc2fi}: $\gain(x^t,i,a) > 0$ 
for all $t_2 < t \leq t_3$.
But this means that $i$ cannot be flipped to $\bar{a}$ in this interval,
and thus $m(t_3) = (i \mapsto a)$ is not a legal flip.
This is a contradiction and so $i \neq k$.
\end{proof}
\begin{proposition}
\label{prop:noTime}
For all $i,j$ and $t_1 < t_2 \leq t_3$, if $(t_1,i \mapsto a) \encourages (t_2, j \mapsto b)$ and\\$(t_1,i \mapsto a) \encourages (t_3, j \mapsto c)$, then $t_2 = t_3$.
\end{proposition}
\begin{proof}
From Proposition~\ref{prop:enc2fi}, we can see that 
for all $t' \in [t_1 + 1,t_3]$, $\gain(x^{t'},j,c) > 0$, so $b = c$ and 
$j$ couldn't have flipped from $c$ to $\bar{c}$ between $t_2$ and $t_3$.
Thus, for $(t_2, j \mapsto c)$ to be a legal flip, we must have $t_2 = t_3$.
\end{proof}
Now, if we trace back along the arrows 
of the encouragement relation, then each flip in the path $p$ 
is the start of a path of encouraged-by links that ends at 
one of the $n$ courageous flips.

One final case to exclude is that there might be two encouragement paths that go in the opposite direction over the same positions.
This cannot happen:
\begin{proposition}
\label{prop:noOpposite}
Having both of the following encouragement paths is impossible:
\begin{align*}
    \bot & \encourages (t_1,i_1 \mapsto b_1) \encourages (t_2,i_2 \mapsto b_2) &  \encourages \cdots  \encourages & \quad (t_m,i_m \mapsto b_m) \\
    \bot  & \encourages (s_m,i_m \mapsto c_m) \encourages (s_{m - 1},i_{m - 1} \mapsto c_{m - 1}) & \encourages \cdots  \encourages & \quad (s_1,i_1 \mapsto c_1)
\end{align*}
\end{proposition}
\begin{proof}
Without loss of generality (by relabeling), we can assume that $t_1 < s_1$.
We can extend this with the following claim:
\\[0.2cm]
\noindent\textbf{Claim:} If $t_k < s_k$ then $t_{k + 1} < s_{k + 1}$

Since $(t_k,i_k \mapsto b_k) \encourages (t_{k + 1},i_{k + 1} \mapsto b_{k+1})$, 
we have, for all $t \in [t_k + 1, t_{k + 1}], \; x^{t}[i_k] = b_k$.
Thus we can't have $i_k$ flipping in that interval, so $s_k > t_{k + 1}$.

But now look at $(s_{k + 1},i_{k + 1} \mapsto c_{k + 1}) \encourages (s_k,i_k \mapsto c_k)$.
This shows that we also have, for all $t' \in [s_{k + 1} + 1, s_k], \; x^{t'}[i_{k + 1}] = c_{k + 1}$.
So for both flips at $i_{k + 1}$ to happen, we need $s_{k + 1} > t_{k + 1}$.
\\[0.2cm]
Applying the claim repeatedly gets us $t_m < s_m$.
But this means that $i_m$ flipped before $m(s_m)$, so by Proposition~\ref{prop:CourStart}, $(s_m,i_m \mapsto c_m)$ could not have been courageous.
\end{proof}
This means that it is sufficient to simply count the number of undirected paths in the encouragement trees.
We now pull all the results together to complete the proof.

\vspace{0.3cm}
\begin{proof}[of \textbf{Theorem~\ref{thm:main}}]
Consider any directed path $P$ in the fitness graph, and its corresponding flip function $m$. 
By the completeness of Definition~\ref{def:cour}, 
we know that every flip must have been either courageous or encouraged.

Any encouraged flip is the end-point of a unique (non-zero length) path 
in the graph of the encouragement relation, starting from some courageous flip,
and hence of some path in the constraint graph
(where Proposition~\ref{prop:noReverse} established that they're encouragement paths, not walks; and Proposition~\ref{prop:noTime} established that the encouragement paths are uniquely determined by the 
sequence of variable positions that they pass through.)
From Proposition~\ref{prop:noOpposite}, we know that there cannot be two encouragement paths that traverse the same positions but in opposite directions.
Thus, there can only be as many non-zero-length encouragement paths 
as undirected paths in our constraint graph.
Since our constraint graph is a tree, an undirected non-zero length path 
is uniquely determined by its pair of endpoints.
Thus, there are at most ${n \choose 2}$ of these paths.

From Proposition~\ref{prop:CourStart}, there are at most $n$ courageous flips (encouragement paths of length $0$).
Thus, our path $P$ must have length at most $n + {n \choose 2}$.
\end{proof}

\section{Long Paths in Landscapes with Simple Constraint Graphs}
\label{sec:robustness}

In this section we show that the conditions in Theorem~\ref{thm:main} are essential.
We exhibit binary VCSP instances with very simple constraint graphs 
where the associated fitness graphs have exponentially-long directed paths,
and hence the performance of local search algorithms on the associated fitness 
landscapes might be extremely poor.
\begin{example}{\bf (Domain size 3)}
\label{ex:counting}
Consider the family of binary VCSP instances on $n+1$ variables with domain $\{0,1,\rhd\}$
illustrated in Figure~\ref{fig:counting},
where the constraint graph of each instance is a path, and all constraints
are defined by matrices of the form:
\begin{equation*}
3 ^{i}\begin{pmatrix}
1 & 2 & 3 \\
2 & 3 & 1 \\
3 & 1 & 2
\end{pmatrix}
\end{equation*}
\begin{figure}[htb]
    \centering
    \includegraphics[width = 0.9\textwidth]{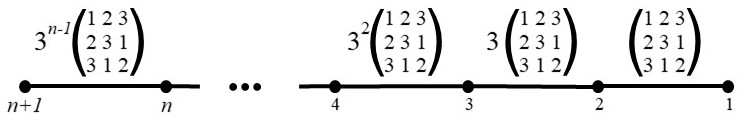}
    \vspace{-0.2cm}
    \caption{A family of binary VCSP instances over domain size 3 
    }
    \label{fig:counting}
\end{figure}

Even though the constraint graph of each instance is just a path of length~$n$, 
we now show that the corresponding fitness graph, contains a directed path whose length
grows exponentially in $n$.

Notice that given two natural numbers $M,M' < 2^n$, written in binary as $x^M$, $x^{M'} \in \{0,1\}^n$ (with the most significant bit on the left and the least significant bit on the right), 
we have that if $M' > M$ then $f(x^{M'}) > f(x^M)$.
Thus, counting up in binary from $0^{n+1}$ to $01^{n}$ is monotonically increasing in fitness.
However, $x^{M + 1}$ is often more than a single flip away from $x^M$ 
(consider the transition from $x^M = 01^{n}$ for an extreme example).
We handle these multi-flip cases with our third domain value, $\rhd$, as follows:
(1) given $x^M = y01^k$ for some $y \in \{0,1\}^{n - k}$, we proceed to replace 
the $1$s in the right-most block of $1$s by $\rhd$, starting from $x^M_{k}$ and moving to the right;
(2) from $y0\rhd^k$ we can take a 1-flip to $y1\rhd^k$ (regardless of whether $y$ ends with $0$ or $1$);
(3) from $x' = y1\rhd^k$, we replace the $\rhd$s by $0$s, starting from the rightmost $\rhd$ (i.e., $x'_1$) and moving to the left. 
Each of these individual 1-flips increases the overall fitness, since the 
fitness gain due to the constraint to the left is greater than the fitness
loss due to the constraint on the right (where there is one).

This lets our sequence of moves count in binary from $0^{n+1}$ to $01^n$,
while using extra steps with $\rhd$s to make sure all transitions are improving 1-flips; 
thus, this directed path in the fitness graph has a length greater than $2^n$. 
\end{example}
Note that although the fitness graphs corresponding to Example~\ref{ex:counting} have long 
improving paths, standard local search algorithms would be unlikely to follow these paths.
However, with careful padding, Example~\ref{ex:counting} has been converted to a 
family of Boolean VCSP instances of treewidth $7$ where even a popular local search algorithm like steepest ascent will follow an exponentially long improving path~\cite{CCKW2020}.

Our final example is a family of binary Boolean VCSP instances where the constraint graph of each instance has tree-width two and maximum degree three,
but the associated fitness graph contains an exponentially long directed path.
This example is a simplified and corrected version of a similar example for 
the {\sc Max-Cut} problem, described by Monien and Tscheuschner~\citeyear{Monien10:degree4}.
Note, however, that by allowing general valued constraints, instead of just 
{\sc Max-Cut} constraints, we are able to reduce the required maximum degree from four to three.  
\begin{example}{\bf (Tree-width 2)}
\label{ex:countBool}
Consider a family of binary Boolean VCSP instances with $n = 4K + 1$ variables.
For each instance the constraint graph contains a sequence of disjoint cycles of length four, 
linked together by a single additional edge joining each consecutive pair of cycles.
The $i$-th cycle (for $0 \leq i \leq K - 1$) has the constraints 
illustrated in Figure~\ref{fig:countBool}
(where the $w_i$ values are defined recursively with $w_0 = 0$).

\begin{figure}[htb]
\begin{center}
\includegraphics[width = 0.8\textwidth]{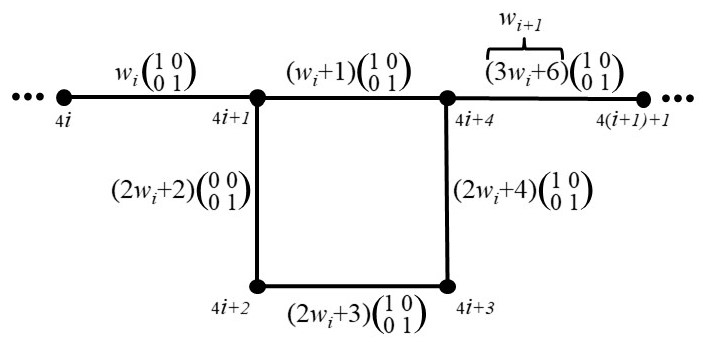}
\end{center}
\vspace{-0.5cm}
\caption{A single cycle from the family of binary Boolean VCSP instances defined in Example~\protect\ref{ex:countBool}}
\label{fig:countBool}
\end{figure}
The final cycle is replaced by a single variable $n$ with unary constraint \uC{1}{-w_{K}}.
Hence the constraint graph of any instance has maximum degree three and treewidth two
as illustrated in Figure~\ref{fig:countBoolPath}
\begin{figure}
    \centering
    \includegraphics[width = \textwidth]{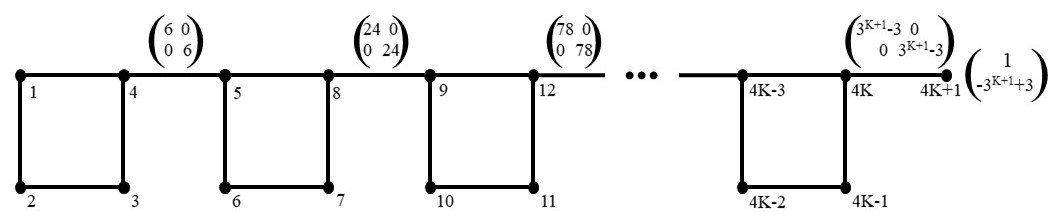}
    \caption{The family of binary Boolean VCSP instances defined in Example~\protect\ref{ex:countBool}}
    \label{fig:countBoolPath}
\end{figure}

To begin the long path all variables are assigned 0, except $x_n = 1$.
The path will proceed by always flipping variables in the smallest $4$-cycle block possible.

Within each $4$-cycle block, we number the variables anti-clockwise, from the top left,
and write the assignment to these $4$ variables ordered by decreasing index 
as $x_{4i + 4}x_{4i + 3}x_{4i + 2}x_{4i + 1}$.
We will make the following transitions within each cycle:
if $x_{4(i + 1) + 1} = 1$ then we'll transition $0000 \rightarrow 1000 \rightarrow 1001 \rightarrow 1101$;
if $x_{4(i + 1) + 1} = 0$ then we'll transition $1101 \rightarrow 0101 \rightarrow 0100 \rightarrow 0110 \rightarrow 0010 \rightarrow 0011 \rightarrow 0001 \rightarrow 0000$.
Every time that $x_{4i + 1}$ is flipped from $0$ to $1$ or vice versa, we'll recurse to the $(i - 1)$th cycle.
Because $x_{4i + 1}$ ends up flipping from $1$ to $0$ twice as often as $x_{4(i + 1) + 1}$, this means that we double the number of flips in each cycle.
Variable $n$ will flip once, from $1$ to $0$, due to the unary constraint, 
which will cause $x_{4(K - 1) + 1}$ to flip twice from $1$ to $0$, 
which will cause $x_{4(K - 2) + 1}$ to flip four times from $1$ to $0$,
and so on,
until eventually this will cause $x_1$ to flip $2^K$ times from $1$ to $0$.
Hence we have an improving path in the fitness graph of length greater than $2^K$.
\end{example}

\section{Applications to Models of Biological Evolution}
\label{sec:bio}

Focusing on the maximum length of improving paths in a fitness graph, rather than the run-time of a particular local search algorithm, lets us use our results in settings where the details of the local search algorithm are unknown or highly contingent.

The most notable example of this is in modelling biological evolution.
In such models, the value of the fitness function is interpreted as a measure of biological fitness (for example, expected number of offspring) and each variable assignment represents the values of the {\em alleles} at a sequence of {\em genetic loci}.
The constraint graph can then be interpreted as a {\em gene-interaction network}.
This gene-interaction network representation of biological fitness landscapes is similar to, but more general than, classic biological models such as 
the NK-model of fitness landscapes~\cite{KL87,KW89,Kaznatcheev2019,strimbu2019}.

The notion of sign-interaction that is central to Section~\ref{sec:signeq} is based on the biological idea of \emph{sign-epistasis} that is central to the analysis of evolutionary dynamics on fitness landscapes~\cite{PKWT07,PSKT11,CGB13,Kaznatcheev2019}.
Our uniqueness results for minimal representations of fitness landscapes 
and fitness graphs can be seen as a way to unambiguously answer which loci have sign-epistasis and to represent the structure of that sign-epistasis.

The process of biological evolution is often viewed as 
some form of randomised uphill climb in a given fitness landscape.
However, the exact probability of any particular adaptive mutation 
arising and fixing in a given population is often 
unknown (or even potentially unknowable in historic cases).
Hence it is very helpful to be able to reason over wide classes of local search algorithms, as we do here.
By showing that for some kinds of fitness landscapes all improving paths are short, 
we have shown that any such randomised algorithm will perform well on such landscapes, regardless of the exact probabilities of particular moves.
For example, we have shown that when the structure of the sign-epistasis is a tree,
then evolution will always reach a local fitness peak in a small number of moves, whatever sequence of improving moves is chosen.

This work has focused on bounding the worst-case complexity of local search algorithms.
Future work could identify further examples of fitness landscapes 
where local peaks can be found efficiently 
by considering less restrictive measures such as randomised complexity.
For randomised complexity, it would be possible to relax the condition that moves must be strictly improving ($f(x^t) < f(x^{t + 1})$) and consider the expected length of walks that include neutral ($f(x^t) = f(x^{t + 1})$) or even deleterious steps ($f(x^t) < f(x^{t + 1})$).
For worst-case complexity, considering neighbourhoods larger than 1-flip can only increase 
the length of the longest improving path. However, for randomised complexity
it would also be possible to consider the effects of larger neighbourhoods
(corresponding to larger mutations), or model the effects of 
recombination and sex.

In a different direction, it is known that there are classes of landscapes 
where locally optimal assignments cannot be found efficiently by any local search algorithm~\cite{Kaznatcheev2019}.
In such cases the structure of the fitness graph can be viewed as an ultimate constraint, that prevents evolution from stabilizing at a local fitness peak~\cite{Kaznatcheev2019}; such cases will give rise to \emph{open-ended evolution}. 
Identifying families of constraint graphs that lead to intractable local search problems therefore corresponds to finding forms of gene-interaction network 
that enable open-ended evolution.

Of course, by treating fitness as a scalar, fitness landscapes are themselves an idealization of the rich multi-faceted concept of biological fitness.
One direction for future work would be to find representations similar to VCSPs but for the richer model of \emph{game landscapes} that account for frequency-dependent fitness~\cite{KazEcoEvo}.

\section{Conclusion}
\label{sec:conc}

In this paper, we have considered the broad class of fitness landscapes that can be modelled by the combined effect of simple interactions of a few variables, where each of these interactions is described by an arbitrary valued constraint.
Modelling fitness landscapes in this way allows us to classify them in new ways:
for example by identifying a minimal constraint graph, and then characterising properties of this constraint graph.

This work raises a number of immediate further questions:
\begin{enumerate}
    \item How can we characterise the fitness landscapes that can be represented by VCSP instances of each fixed arity?
    Since the results on magnitude-equivalence (Theorems~\ref{thm:simplerep} and \ref{thm:valmin}) are based on representations of pseudo-boolean functions, they are easy to generalise to higher arity.
    But generalising the results on sign-equivalence (Theorems~\ref{thm:trimprep} and ~\ref{thm:signmin}) to higher arity will require new techniques.
    
    \item We have shown that when a family of fitness landscapes can be represented by binary Boolean VCSP instances where the constraint graph has maximum degree 2, then the minimal span of a sign-equivalent instance grows only polynomially (Example~\ref{ex:cycle}), and hence  
    finding a local optimum in the corresponding fitness landscapes by \emph{any} local search algorithm takes only polynomial time (Proposition~\ref{prop:spanPath}). 
    
    Can this result be generalised to wider classes of landscapes? This may be difficult:  
    we have shown that even for landscapes representable by binary Boolean VCSP instances  with tree-structured constraint graphs, the minimal span may grow exponentially with the number
    of variables (Example~\ref{ex:bigSpanTree}).
    
    \item We have shown that when a family of fitness landscapes can be represented by binary Boolean VCSP instances where the constraint graph is tree-structured, then finding a local optimum 
    in the corresponding fitness landscapes by any local search algorithm takes only polynomial time (Theorem~\ref{thm:main}). 
    
    Can this result be generalised to wider classes of landscapes? 
    This may be difficult:  
    we have shown examples over a slightly larger domain, or allowing slightly more general constraint graphs, where some local search algorithms can take exponential time to find a local optimum (Examples~\ref{ex:counting} and \ref{ex:countBool}).

    An alternative approach is to move away from considering just topological features of constraint graphs and look at features of the valued constraints themselves.
    Restricting the kinds of constraints (for example, considering only constraints that are convex, or sub-modular) could allow us to identify additional 
    classes of VCSP instances that are tractable for local search.
\end{enumerate}
We believe that the tools for classifying fitness landscapes that we have begun to develop here will allow considerable further progress, and may eventually help to shed more light on the question of why local search algorithms can be extremely effective in practice, for many kinds of optimisation problems.

\section*{Acknowledgments}\label{sec:Acknowledgments}
David A. Cohen was supported by Leverhulme Trust Grant RPG-2018-161.
Artem Kaznatcheev was supported by the Theory Division at the Department of Translational Hematology and Oncology Research, Cleveland Clinic.
An earlier version of some parts of this paper was presented
at the {\it 25th International Conference on Principles and Practice of Constraint Programming}~\cite{DBLP:conf/cp/KaznatcheevCJ19}.

\bibliographystyle{apacite}
\bibliography{MasterCSP}

\begin{thebibliography}{}

\bibitem [\protect \citeauthoryear {%
Aaronson%
}{%
Aaronson%
}{%
{\protect \APACyear {2006}}%
}]{%
Aaronson2006}
\APACinsertmetastar {%
Aaronson2006}%
\begin{APACrefauthors}%
Aaronson, S.%
\end{APACrefauthors}%
\unskip\
\newblock
\APACrefYearMonthDay{2006}{}{}.
\newblock
{\BBOQ}\APACrefatitle {Lower Bounds for Local Search by Quantum Arguments}
  {Lower bounds for local search by quantum arguments}.{\BBCQ}
\newblock
\APACjournalVolNumPages{SIAM Journal on Computing}{35}{4}{804--824}.
\PrintBackRefs{\CurrentBib}

\bibitem [\protect \citeauthoryear {%
Carbonnel%
, Romero%
\BCBL {}\ \BBA {} Zivny%
}{%
Carbonnel%
\ \protect \BOthers {.}}{%
{\protect \APACyear {2018}}%
}]{%
Carbonnel2018}
\APACinsertmetastar {%
Carbonnel2018}%
\begin{APACrefauthors}%
Carbonnel, C.%
, Romero, M.%
\BCBL {}\ \BBA {} Zivny, S.%
\end{APACrefauthors}%
\unskip\
\newblock
\APACrefYearMonthDay{2018}{}{}.
\newblock
{\BBOQ}\APACrefatitle {The Complexity of General-Valued {CSP}s Seen from the
  Other Side} {The complexity of general-valued {CSP}s seen from the other
  side}.{\BBCQ}
\newblock
\BIn{} \APACrefbtitle {59th {IEEE} {A}nnual {S}ymposium on {F}oundations of
  {C}omputer {S}cience, {FOCS} 2018} {59th {IEEE} {A}nnual {S}ymposium on
  {F}oundations of {C}omputer {S}cience, {FOCS} 2018}\ (\BPGS\ 236--246).
\PrintBackRefs{\CurrentBib}

\bibitem [\protect \citeauthoryear {%
Chapdelaine%
\ \BBA {} Creignou%
}{%
Chapdelaine%
\ \BBA {} Creignou%
}{%
{\protect \APACyear {2005}}%
}]{%
Chapdelaine2005}
\APACinsertmetastar {%
Chapdelaine2005}%
\begin{APACrefauthors}%
Chapdelaine, P.%
\BCBT {}\ \BBA {} Creignou, N.%
\end{APACrefauthors}%
\unskip\
\newblock
\APACrefYearMonthDay{2005}{}{}.
\newblock
{\BBOQ}\APACrefatitle {The Complexity of {B}oolean Constraint Satisfaction
  Local Search Problems} {The complexity of {B}oolean constraint satisfaction
  local search problems}.{\BBCQ}
\newblock
\APACjournalVolNumPages{Annals of Mathematics and Artificial
  Intelligence}{43}{1-4}{51--63}.
\PrintBackRefs{\CurrentBib}

\bibitem [\protect \citeauthoryear {%
Cohen%
, Cooper%
, Creed%
, Jeavons%
\BCBL {}\ \BBA {} Zivny%
}{%
Cohen%
\ \protect \BOthers {.}}{%
{\protect \APACyear {2013}}%
}]{%
Cohen2013algebraic}
\APACinsertmetastar {%
Cohen2013algebraic}%
\begin{APACrefauthors}%
Cohen, D\BPBI A.%
, Cooper, M\BPBI C.%
, Creed, P.%
, Jeavons, P\BPBI G.%
\BCBL {}\ \BBA {} Zivny, S.%
\end{APACrefauthors}%
\unskip\
\newblock
\APACrefYearMonthDay{2013}{}{}.
\newblock
{\BBOQ}\APACrefatitle {An Algebraic Theory of Complexity for Discrete
  Optimization} {An algebraic theory of complexity for discrete
  optimization}.{\BBCQ}
\newblock
\APACjournalVolNumPages{{SIAM} J. Comput.}{42}{5}{1915--1939}.
\PrintBackRefs{\CurrentBib}

\bibitem [\protect \citeauthoryear {%
Cohen%
, Cooper%
, Kaznatcheev%
\BCBL {}\ \BBA {} Wallace%
}{%
Cohen%
\ \protect \BOthers {.}}{%
{\protect \APACyear {2020}}%
}]{%
CCKW2020}
\APACinsertmetastar {%
CCKW2020}%
\begin{APACrefauthors}%
Cohen, D\BPBI A.%
, Cooper, M\BPBI C.%
, Kaznatcheev, A.%
\BCBL {}\ \BBA {} Wallace, M.%
\end{APACrefauthors}%
\unskip\
\newblock
\APACrefYearMonthDay{2020}{}{}.
\newblock
{\BBOQ}\APACrefatitle {Steepest ascent can be exponential in bounded treewidth
  problems} {Steepest ascent can be exponential in bounded treewidth
  problems}.{\BBCQ}
\newblock
\APACjournalVolNumPages{Operations Research Letters}{48}{}{217-224}.
\PrintBackRefs{\CurrentBib}

\bibitem [\protect \citeauthoryear {%
Cooper%
, De~Givry%
\BCBL {}\ \BBA {} Schiex%
}{%
Cooper%
\ \protect \BOthers {.}}{%
{\protect \APACyear {2007}}%
}]{%
Cooper07:SAC}
\APACinsertmetastar {%
Cooper07:SAC}%
\begin{APACrefauthors}%
Cooper, M\BPBI C.%
, De~Givry, S.%
\BCBL {}\ \BBA {} Schiex, T.%
\end{APACrefauthors}%
\unskip\
\newblock
\APACrefYearMonthDay{2007}{}{}.
\newblock
{\BBOQ}\APACrefatitle {Optimal Soft Arc Consistency} {Optimal soft arc
  consistency}.{\BBCQ}
\newblock
\BIn{} \APACrefbtitle {Proceedings of the 20th {I}nternational {J}oint
  {C}onference on {A}rtificial {I}ntelligence {IJCAI'07}} {Proceedings of the
  20th {I}nternational {J}oint {C}onference on {A}rtificial {I}ntelligence
  {IJCAI'07}}\ (\BPGS\ 68--73).
\PrintBackRefs{\CurrentBib}

\bibitem [\protect \citeauthoryear {%
Crama%
\ \BBA {} Hammer%
}{%
Crama%
\ \BBA {} Hammer%
}{%
{\protect \APACyear {2011}}%
}]{%
CramaHammer2011}
\APACinsertmetastar {%
CramaHammer2011}%
\begin{APACrefauthors}%
Crama, Y.%
\BCBT {}\ \BBA {} Hammer, P.%
\end{APACrefauthors}%
\unskip\
\newblock
\APACrefYear{2011}.
\newblock
\APACrefbtitle {Boolean Functions: Theory, Algorithms, and Applications}
  {Boolean functions: Theory, algorithms, and applications}.
\newblock
\APACaddressPublisher{}{Cambridge University Press}.
\PrintBackRefs{\CurrentBib}

\bibitem [\protect \citeauthoryear {%
Crona%
, Greene%
\BCBL {}\ \BBA {} Barlow%
}{%
Crona%
\ \protect \BOthers {.}}{%
{\protect \APACyear {2013}}%
}]{%
CGB13}
\APACinsertmetastar {%
CGB13}%
\begin{APACrefauthors}%
Crona, K.%
, Greene, D.%
\BCBL {}\ \BBA {} Barlow, M.%
\end{APACrefauthors}%
\unskip\
\newblock
\APACrefYearMonthDay{2013}{}{}.
\newblock
{\BBOQ}\APACrefatitle {The peaks and geometry of fitness landscapes.} {The
  peaks and geometry of fitness landscapes.}{\BBCQ}
\newblock
\APACjournalVolNumPages{Journal of Theoretical Biology}{317}{}{1-10}.
\PrintBackRefs{\CurrentBib}

\bibitem [\protect \citeauthoryear {%
{de Visser}%
, Park%
\BCBL {}\ \BBA {} Krug%
}{%
{de Visser}%
\ \protect \BOthers {.}}{%
{\protect \APACyear {2009}}%
}]{%
dVPK09}
\APACinsertmetastar {%
dVPK09}%
\begin{APACrefauthors}%
{de Visser}, J.%
, Park, S.%
\BCBL {}\ \BBA {} Krug, J.%
\end{APACrefauthors}%
\unskip\
\newblock
\APACrefYearMonthDay{2009}{}{}.
\newblock
{\BBOQ}\APACrefatitle {Exploring the effect of sex on empirical fitness
  landscapes.} {Exploring the effect of sex on empirical fitness
  landscapes.}{\BBCQ}
\newblock
\APACjournalVolNumPages{The American Naturalist}{174 Supplement 1}{}{S15-30}.
\PrintBackRefs{\CurrentBib}

\bibitem [\protect \citeauthoryear {%
F{\"{a}}rnqvist%
}{%
F{\"{a}}rnqvist%
}{%
{\protect \APACyear {2012}}%
}]{%
Farnqvist12}
\APACinsertmetastar {%
Farnqvist12}%
\begin{APACrefauthors}%
F{\"{a}}rnqvist, T.%
\end{APACrefauthors}%
\unskip\
\newblock
\APACrefYearMonthDay{2012}{}{}.
\newblock
{\BBOQ}\APACrefatitle {Constraint Optimization Problems and Bounded Tree-Width
  Revisited} {Constraint optimization problems and bounded tree-width
  revisited}.{\BBCQ}
\newblock
\BIn{} \APACrefbtitle {Integration of {AI} and {OR} Techniques in Contraint
  Programming for Combinatorial Optimzation Problems {CPAIOR}} {Integration of
  {AI} and {OR} techniques in contraint programming for combinatorial
  optimzation problems {CPAIOR}}\ (\BVOL\ {LNCS} 7298, \BPGS\ 163--179).
\newblock
\APACaddressPublisher{}{Springer}.
\PrintBackRefs{\CurrentBib}

\bibitem [\protect \citeauthoryear {%
Garey%
\ \BBA {} Johnson%
}{%
Garey%
\ \BBA {} Johnson%
}{%
{\protect \APACyear {1979}}%
}]{%
Garey1979}
\APACinsertmetastar {%
Garey1979}%
\begin{APACrefauthors}%
Garey, M.%
\BCBT {}\ \BBA {} Johnson, D.%
\end{APACrefauthors}%
\unskip\
\newblock
\APACrefYear{1979}.
\newblock
\APACrefbtitle {Computers and Intractability: A Guide to the Theory of
  {NP}-Completeness} {Computers and intractability: A guide to the theory of
  {NP}-completeness}.
\newblock
\APACaddressPublisher{San Francisco, CA.}{Freeman}.
\PrintBackRefs{\CurrentBib}

\bibitem [\protect \citeauthoryear {%
Johnson%
, Papadimitriou%
\BCBL {}\ \BBA {} Yannakakis%
}{%
Johnson%
\ \protect \BOthers {.}}{%
{\protect \APACyear {1988}}%
}]{%
PLS}
\APACinsertmetastar {%
PLS}%
\begin{APACrefauthors}%
Johnson, D.%
, Papadimitriou, C.%
\BCBL {}\ \BBA {} Yannakakis, M.%
\end{APACrefauthors}%
\unskip\
\newblock
\APACrefYearMonthDay{1988}{}{}.
\newblock
{\BBOQ}\APACrefatitle {How easy is local search?} {How easy is local
  search?}{\BBCQ}
\newblock
\APACjournalVolNumPages{Journal of Computer and System Sciences}{37}{}{79-100}.
\PrintBackRefs{\CurrentBib}

\bibitem [\protect \citeauthoryear {%
Kauffman%
\ \BBA {} Levin%
}{%
Kauffman%
\ \BBA {} Levin%
}{%
{\protect \APACyear {1987}}%
}]{%
KL87}
\APACinsertmetastar {%
KL87}%
\begin{APACrefauthors}%
Kauffman, S.%
\BCBT {}\ \BBA {} Levin, S.%
\end{APACrefauthors}%
\unskip\
\newblock
\APACrefYearMonthDay{1987}{}{}.
\newblock
{\BBOQ}\APACrefatitle {Towards a general theory of adaptive walks on rugged
  landscapes.} {Towards a general theory of adaptive walks on rugged
  landscapes.}{\BBCQ}
\newblock
\APACjournalVolNumPages{Journal of Theoretical Biology}{128}{}{11-45}.
\PrintBackRefs{\CurrentBib}

\bibitem [\protect \citeauthoryear {%
Kauffman%
\ \BBA {} Weinberger%
}{%
Kauffman%
\ \BBA {} Weinberger%
}{%
{\protect \APACyear {1989}}%
}]{%
KW89}
\APACinsertmetastar {%
KW89}%
\begin{APACrefauthors}%
Kauffman, S.%
\BCBT {}\ \BBA {} Weinberger, E.%
\end{APACrefauthors}%
\unskip\
\newblock
\APACrefYearMonthDay{1989}{}{}.
\newblock
{\BBOQ}\APACrefatitle {The {NK} model of rugged fitness landscapes and its
  application to maturation of the immune response.} {The {NK} model of rugged
  fitness landscapes and its application to maturation of the immune
  response.}{\BBCQ}
\newblock
\APACjournalVolNumPages{Journal of Theoretical Biology}{141}{}{211-245}.
\PrintBackRefs{\CurrentBib}

\bibitem [\protect \citeauthoryear {%
Kaznatcheev%
}{%
Kaznatcheev%
}{%
{\protect \APACyear {2019}}%
}]{%
Kaznatcheev2019}
\APACinsertmetastar {%
Kaznatcheev2019}%
\begin{APACrefauthors}%
Kaznatcheev, A.%
\end{APACrefauthors}%
\unskip\
\newblock
\APACrefYearMonthDay{2019}{}{}.
\newblock
{\BBOQ}\APACrefatitle {Computational Complexity as an Ultimate Constraint on
  Evolution} {Computational complexity as an ultimate constraint on
  evolution}.{\BBCQ}
\newblock
\APACjournalVolNumPages{Genetics}{212}{1}{245--265}.
\PrintBackRefs{\CurrentBib}

\bibitem [\protect \citeauthoryear {%
Kaznatcheev%
}{%
Kaznatcheev%
}{%
{\protect \APACyear {2020}}%
{\protect \APACexlab {{\protect \BCnt {1}}}}}]{%
KazThesis}
\APACinsertmetastar {%
KazThesis}%
\begin{APACrefauthors}%
Kaznatcheev, A.%
\end{APACrefauthors}%
\unskip\
\newblock
\APACrefYearMonthDay{2020{\protect \BCnt {1}}}{}{}.
\newblock
\APACrefbtitle {{Algorithmic Biology of Evolution and Ecology}.} {{Algorithmic
  Biology of Evolution and Ecology}.}
\newblock
\APACaddressPublisher{}{University of Oxford}.
\PrintBackRefs{\CurrentBib}

\bibitem [\protect \citeauthoryear {%
Kaznatcheev%
}{%
Kaznatcheev%
}{%
{\protect \APACyear {2020}}%
{\protect \APACexlab {{\protect \BCnt {2}}}}}]{%
KazEcoEvo}
\APACinsertmetastar {%
KazEcoEvo}%
\begin{APACrefauthors}%
Kaznatcheev, A.%
\end{APACrefauthors}%
\unskip\
\newblock
\APACrefYearMonthDay{2020{\protect \BCnt {2}}}{}{}.
\newblock
{\BBOQ}\APACrefatitle {Evolution is exponentially more powerful with
  frequency-dependent selection} {Evolution is exponentially more powerful with
  frequency-dependent selection}.{\BBCQ}
\newblock
\APACjournalVolNumPages{bioRxiv}{}{}{}.
\newblock
\begin{APACrefDOI} \doi{10.1101/2020.05.03.075069} \end{APACrefDOI}
\PrintBackRefs{\CurrentBib}

\bibitem [\protect \citeauthoryear {%
Kaznatcheev%
, Cohen%
\BCBL {}\ \BBA {} Jeavons%
}{%
Kaznatcheev%
\ \protect \BOthers {.}}{%
{\protect \APACyear {2019}}%
}]{%
DBLP:conf/cp/KaznatcheevCJ19}
\APACinsertmetastar {%
DBLP:conf/cp/KaznatcheevCJ19}%
\begin{APACrefauthors}%
Kaznatcheev, A.%
, Cohen, D\BPBI A.%
\BCBL {}\ \BBA {} Jeavons, P\BPBI G.%
\end{APACrefauthors}%
\unskip\
\newblock
\APACrefYearMonthDay{2019}{}{}.
\newblock
{\BBOQ}\APACrefatitle {Representing Fitness Landscapes by Valued Constraints to
  Understand the Complexity of Local Search} {Representing fitness landscapes
  by valued constraints to understand the complexity of local search}.{\BBCQ}
\newblock
\BIn{} T.~Schiex\ \BBA {} S.~de Givry\ (\BEDS), \APACrefbtitle {International
  {C}onference on {P}rinciples and {P}ractice of {C}onstraint {P}rogramming,
  {CP} 2019} {International {C}onference on {P}rinciples and {P}ractice of
  {C}onstraint {P}rogramming, {CP} 2019}\ (\BVOL\ {LNCS} 11802, \BPGS\
  300--316).
\newblock
\APACaddressPublisher{}{Springer}.
\PrintBackRefs{\CurrentBib}

\bibitem [\protect \citeauthoryear {%
Kolmogorov%
\ \BBA {} Zivny%
}{%
Kolmogorov%
\ \BBA {} Zivny%
}{%
{\protect \APACyear {2013}}%
}]{%
Kolmogorov2013}
\APACinsertmetastar {%
Kolmogorov2013}%
\begin{APACrefauthors}%
Kolmogorov, V.%
\BCBT {}\ \BBA {} Zivny, S.%
\end{APACrefauthors}%
\unskip\
\newblock
\APACrefYearMonthDay{2013}{}{}.
\newblock
{\BBOQ}\APACrefatitle {The complexity of conservative valued {CSP}s} {The
  complexity of conservative valued {CSP}s}.{\BBCQ}
\newblock
\APACjournalVolNumPages{Journal of the {ACM}}{60}{2}{10:1--10:38}.
\PrintBackRefs{\CurrentBib}

\bibitem [\protect \citeauthoryear {%
Llewellyn%
, Tovey%
\BCBL {}\ \BBA {} Trick%
}{%
Llewellyn%
\ \protect \BOthers {.}}{%
{\protect \APACyear {1989}}%
}]{%
Llewellyn1989}
\APACinsertmetastar {%
Llewellyn1989}%
\begin{APACrefauthors}%
Llewellyn, D\BPBI C.%
, Tovey, C\BPBI A.%
\BCBL {}\ \BBA {} Trick, M\BPBI A.%
\end{APACrefauthors}%
\unskip\
\newblock
\APACrefYearMonthDay{1989}{}{}.
\newblock
{\BBOQ}\APACrefatitle {Local optimization on graphs} {Local optimization on
  graphs}.{\BBCQ}
\newblock
\APACjournalVolNumPages{Discrete Applied Mathematics}{23}{2}{157--178}.
\PrintBackRefs{\CurrentBib}

\bibitem [\protect \citeauthoryear {%
Malan%
\ \BBA {} Engelbrecht%
}{%
Malan%
\ \BBA {} Engelbrecht%
}{%
{\protect \APACyear {2013}}%
}]{%
Malan2013}
\APACinsertmetastar {%
Malan2013}%
\begin{APACrefauthors}%
Malan, K\BPBI M.%
\BCBT {}\ \BBA {} Engelbrecht, A\BPBI P.%
\end{APACrefauthors}%
\unskip\
\newblock
\APACrefYearMonthDay{2013}{}{}.
\newblock
{\BBOQ}\APACrefatitle {A survey of techniques for characterising fitness
  landscapes and some possible ways forward} {A survey of techniques for
  characterising fitness landscapes and some possible ways forward}.{\BBCQ}
\newblock
\APACjournalVolNumPages{Information Sciences}{241}{}{148 - 163}.
\PrintBackRefs{\CurrentBib}

\bibitem [\protect \citeauthoryear {%
Monien%
\ \BBA {} Tscheuschner%
}{%
Monien%
\ \BBA {} Tscheuschner%
}{%
{\protect \APACyear {2010}}%
}]{%
Monien10:degree4}
\APACinsertmetastar {%
Monien10:degree4}%
\begin{APACrefauthors}%
Monien, B.%
\BCBT {}\ \BBA {} Tscheuschner, T.%
\end{APACrefauthors}%
\unskip\
\newblock
\APACrefYearMonthDay{2010}{}{}.
\newblock
{\BBOQ}\APACrefatitle {On the Power of Nodes of Degree Four in the Local
  Max-Cut Problem} {On the power of nodes of degree four in the local max-cut
  problem}.{\BBCQ}
\newblock
\BIn{} T.~Calamoneri\ \BBA {} J.~Diaz\ (\BEDS), \APACrefbtitle {Algorithms and
  {C}omplexity} {Algorithms and {C}omplexity}\ (\BPGS\ 264--275).
\newblock
\APACaddressPublisher{}{Springer Berlin Heidelberg}.
\PrintBackRefs{\CurrentBib}

\bibitem [\protect \citeauthoryear {%
Ochoa%
\ \BBA {} Veerapen%
}{%
Ochoa%
\ \BBA {} Veerapen%
}{%
{\protect \APACyear {2018}}%
}]{%
Ochoa2018}
\APACinsertmetastar {%
Ochoa2018}%
\begin{APACrefauthors}%
Ochoa, G.%
\BCBT {}\ \BBA {} Veerapen, N.%
\end{APACrefauthors}%
\unskip\
\newblock
\APACrefYearMonthDay{2018}{}{}.
\newblock
{\BBOQ}\APACrefatitle {Mapping the global structure of {TSP} fitness
  landscapes} {Mapping the global structure of {TSP} fitness
  landscapes}.{\BBCQ}
\newblock
\APACjournalVolNumPages{Journal of Heuristics}{24}{3}{265--294}.
\PrintBackRefs{\CurrentBib}

\bibitem [\protect \citeauthoryear {%
Poelwijk%
, Kiviet%
, Weinreich%
\BCBL {}\ \BBA {} Tans%
}{%
Poelwijk%
\ \protect \BOthers {.}}{%
{\protect \APACyear {2007}}%
}]{%
PKWT07}
\APACinsertmetastar {%
PKWT07}%
\begin{APACrefauthors}%
Poelwijk, F.%
, Kiviet, D.%
, Weinreich, D.%
\BCBL {}\ \BBA {} Tans, S.%
\end{APACrefauthors}%
\unskip\
\newblock
\APACrefYearMonthDay{2007}{}{}.
\newblock
{\BBOQ}\APACrefatitle {Empirical fitness landscapes reveal accessible
  evolutionary paths.} {Empirical fitness landscapes reveal accessible
  evolutionary paths.}{\BBCQ}
\newblock
\APACjournalVolNumPages{Nature}{445}{}{383-386}.
\PrintBackRefs{\CurrentBib}

\bibitem [\protect \citeauthoryear {%
Poelwijk%
, Sorin%
, Kiviet%
\BCBL {}\ \BBA {} Tans%
}{%
Poelwijk%
\ \protect \BOthers {.}}{%
{\protect \APACyear {2011}}%
}]{%
PSKT11}
\APACinsertmetastar {%
PSKT11}%
\begin{APACrefauthors}%
Poelwijk, F.%
, Sorin, T\BHBI N.%
, Kiviet, D.%
\BCBL {}\ \BBA {} Tans, S.%
\end{APACrefauthors}%
\unskip\
\newblock
\APACrefYearMonthDay{2011}{}{}.
\newblock
{\BBOQ}\APACrefatitle {Reciprocal sign epistasis is a necessary condition for
  multi-peaked fitness landscapes.} {Reciprocal sign epistasis is a necessary
  condition for multi-peaked fitness landscapes.}{\BBCQ}
\newblock
\APACjournalVolNumPages{Journal of Theoretical Biology}{272}{}{141 -- 144}.
\PrintBackRefs{\CurrentBib}

\bibitem [\protect \citeauthoryear {%
Schaffer%
\ \BBA {} Yannakakis%
}{%
Schaffer%
\ \BBA {} Yannakakis%
}{%
{\protect \APACyear {1991}}%
}]{%
SY91}
\APACinsertmetastar {%
SY91}%
\begin{APACrefauthors}%
Schaffer, A.%
\BCBT {}\ \BBA {} Yannakakis, M.%
\end{APACrefauthors}%
\unskip\
\newblock
\APACrefYearMonthDay{1991}{}{}.
\newblock
{\BBOQ}\APACrefatitle {Simple local search problems that are hard to solve.}
  {Simple local search problems that are hard to solve.}{\BBCQ}
\newblock
\APACjournalVolNumPages{{SIAM} Journal on Computing}{20}{1}{56--87}.
\PrintBackRefs{\CurrentBib}

\bibitem [\protect \citeauthoryear {%
Strimbu%
}{%
Strimbu%
}{%
{\protect \APACyear {2019}}%
}]{%
strimbu2019}
\APACinsertmetastar {%
strimbu2019}%
\begin{APACrefauthors}%
Strimbu, A.%
\end{APACrefauthors}%
\unskip\
\newblock
\APACrefYearMonthDay{2019}{}{}.
\newblock
{\BBOQ}\APACrefatitle {Simulating Evolution on Fitness Landscapes represented
  by Valued Constraint Satisfaction Problems} {Simulating evolution on fitness
  landscapes represented by valued constraint satisfaction problems}.{\BBCQ}
\newblock
\APACjournalVolNumPages{arXiv:1912.02134}{}{}{}.
\PrintBackRefs{\CurrentBib}

\bibitem [\protect \citeauthoryear {%
Tayarani{-}Najaran%
\ \BBA {} Pr{\"{u}}gel{-}Bennett%
}{%
Tayarani{-}Najaran%
\ \BBA {} Pr{\"{u}}gel{-}Bennett%
}{%
{\protect \APACyear {2014}}%
}]{%
Tayarani-Najaran2014}
\APACinsertmetastar {%
Tayarani-Najaran2014}%
\begin{APACrefauthors}%
Tayarani{-}Najaran, M.%
\BCBT {}\ \BBA {} Pr{\"{u}}gel{-}Bennett, A.%
\end{APACrefauthors}%
\unskip\
\newblock
\APACrefYearMonthDay{2014}{}{}.
\newblock
{\BBOQ}\APACrefatitle {On the Landscape of Combinatorial Optimization Problems}
  {On the landscape of combinatorial optimization problems}.{\BBCQ}
\newblock
\APACjournalVolNumPages{{IEEE} Trans. Evolutionary
  Computation}{18}{3}{420--434}.
\PrintBackRefs{\CurrentBib}

\bibitem [\protect \citeauthoryear {%
Thapper%
\ \BBA {} Zivny%
}{%
Thapper%
\ \BBA {} Zivny%
}{%
{\protect \APACyear {2015}}%
}]{%
Thapper2015}
\APACinsertmetastar {%
Thapper2015}%
\begin{APACrefauthors}%
Thapper, J.%
\BCBT {}\ \BBA {} Zivny, S.%
\end{APACrefauthors}%
\unskip\
\newblock
\APACrefYearMonthDay{2015}{}{}.
\newblock
{\BBOQ}\APACrefatitle {Necessary Conditions for Tractability of Valued {CSP}s}
  {Necessary conditions for tractability of valued {CSP}s}.{\BBCQ}
\newblock
\APACjournalVolNumPages{{SIAM} Journal on Discrete
  Mathematics}{29}{4}{2361--2384}.
\PrintBackRefs{\CurrentBib}

\bibitem [\protect \citeauthoryear {%
Thapper%
\ \BBA {} Zivny%
}{%
Thapper%
\ \BBA {} Zivny%
}{%
{\protect \APACyear {2016}}%
}]{%
Thapper2016}
\APACinsertmetastar {%
Thapper2016}%
\begin{APACrefauthors}%
Thapper, J.%
\BCBT {}\ \BBA {} Zivny, S.%
\end{APACrefauthors}%
\unskip\
\newblock
\APACrefYearMonthDay{2016}{}{}.
\newblock
{\BBOQ}\APACrefatitle {The Complexity of Finite-Valued {CSP}s} {The complexity
  of finite-valued {CSP}s}.{\BBCQ}
\newblock
\APACjournalVolNumPages{Journal of the {ACM}}{63}{4}{37:1--37:33}.
\PrintBackRefs{\CurrentBib}

\bibitem [\protect \citeauthoryear {%
Wright%
}{%
Wright%
}{%
{\protect \APACyear {1932}}%
}]{%
Wright1932}
\APACinsertmetastar {%
Wright1932}%
\begin{APACrefauthors}%
Wright, S.%
\end{APACrefauthors}%
\unskip\
\newblock
\APACrefYearMonthDay{1932}{}{}.
\newblock
{\BBOQ}\APACrefatitle {The roles of mutation, inbreeding, crossbreeding, and
  selection in evolution} {The roles of mutation, inbreeding, crossbreeding,
  and selection in evolution}.{\BBCQ}
\newblock
\BIn{} \APACrefbtitle {{Proc. of the 6th International Congress on Genetics}}
  {{Proc. of the 6th International Congress on Genetics}}\ (\BPGS\ 355--366).
\PrintBackRefs{\CurrentBib}

\end{thebibliography}

\end{document}